%% file: main_v3.tex
\documentclass[10 pt, conference]{IEEEtran}

\IEEEoverridecommandlockouts                              
\usepackage{graphicx}
\usepackage{subcaption}
\usepackage{bbm}
\usepackage{upgreek}
\usepackage{bbold}
\usepackage{color}
\usepackage{graphics} 
\input{my_sections.tex}

\input{mysymbol.sty}
\usepackage{cite}
\usepackage{amsmath}
\usepackage{amsthm}
\usepackage{mathtools}
\usepackage{multirow}
\usepackage{mathabx}
\usepackage{url}
\usepackage{graphics, times}
\usepackage{tikz, epic,eepic}
\usetikzlibrary{shapes,arrows}
\usepackage{pgfplots}
\usepackage{color}
\usepackage{hyperref}
\usepackage[capitalise]{cleveref}
\usepackage{epstopdf}
\usepackage{epsfig, amsbsy}
\usepackage{latexsym}
\usepackage{amscd, verbatim}
\usepackage{multirow}
\usepackage{booktabs}
\usepackage{algorithm}
\usepackage{algorithmic}

\renewcommand{\baselinestretch}{0.97}\selectfont
\newtheorem{theorem}{Theorem}

\newtheorem{proposition}{Proposition}
\newtheorem{corollary}{Corollary}
\newtheorem{lemma}{Lemma}
\newtheorem{definition}{Definition}

{\itshape}{\rmfamily}
\newtheorem{assumption}{Assumption}
\newtheorem{remark}{Remark}

\DeclarePairedDelimiter\floor{\lfloor}{\rfloor}

\def\forall{\text{for all\ }}

\def\l{\lambda}
\def\a{\alpha}
\def\b{\beta}

\newcommand{\aij}{\alpha_{ij}}
\newcommand{\bij}{\beta_{ij}}
\newcommand{\Ni}{\mathcal{N}_i}
\newcommand{\Nil}{\mathcal{N}_i^{\mathcal{L}}}
\newcommand{\Nim}{\mathcal{N}_i^{\mathcal{M}}}
\newcommand{\xl}{x_{\mathcal{L}}}
\newcommand{\Txl}{\Tilde{x}_{\mathcal{L}}}
\newcommand{\xm}{x_{\mathcal{M}}}
\newcommand{\wl}{W_{\mathcal{L}}}
\newcommand{\barwl}{\widebar{W}_{\mathcal{L}}}
\newcommand{\wm}{W_{\mathcal{M}}}

\def\reals{\mathbb{R}}

\newcommand{\oea}[1]{\textcolor{cyan}{#1}}

\newcounter{sideremark}

\def\LL{\mathcal{L}}
\def\M{\mathcal{M}}

\title{\LARGE 
Multi-Agent Trustworthy Consensus under Random Dynamic Attacks }
\author{Orhan Eren Akg\"un$^*$, Sarper Ayd{\i}n$^*$, Stephanie Gil, and Angelia Nedi\'c %
\thanks{O.\ E.\ Akg\"un, S.\ Ayd{\i}n,  and S.\ Gil are with the School of Engineering and Applied Sciences, Harvard
University, Cambridge, MA 02138. E-mail: {\tt\small erenakgun@g.harvard.edu, saydin@seas.harvard.edu, sgil@seas.harvard.edu}. A.\ Nedi\'c are with the School of Electrical, Computer and Energy Engineering, Arizona State University, Tempe, AZ 85281. E-mail:{\tt\small \; angelia.nedich@asu.edu.} $^*$Indicates equal contribution.  
This work has been supported by the NSF awards CNS-2147641, CNS-2147694, and the Defense Advanced Research Projects Agency (DARPA) under Grant No. D24AP00319-00. Approved for public release; distribution is unlimited. The project was or is sponsored by DARPA, and the content does not necessarily reflect the position or policy of the Government, nor should official endorsement be inferred.}}

\begin{document}
\normalsize
\maketitle

%
\begin{abstract}

In this work, we study the consensus problem in which legitimate agents send their values over an undirected communication network in the presence of an unknown subset of malicious or faulty agents. In contrast to former works, we generalize and characterize the properties of consensus dynamics with dependent sequences of malicious transmissions with dynamic (time-varying) rates, based on not necessarily independent trust observations. We consider a detection algorithm utilizing stochastic trust observations available to legitimate agents. Under these conditions, legitimate agents almost surely classify their neighbors and form their trusted neighborhoods correctly with decaying misclassification probabilities. We further prove that the consensus process converges almost surely despite the existence of malicious agents. For a given value of failure probability, we characterize the deviation from the nominal consensus value ideally occurring when there are no malicious agents in the system. We also examine the convergence rate of the process in finite time. Numerical simulations show the convergence among agents and indicate the deviation under different attack scenarios.
\end{abstract}

%

%
\section{Introduction}

\input{introduction.tex}
\section{Consensus Dynamics with Malicious Agents}\label{sec:model}
In this section, we formally introduce the problem. In \Cref{sec:consensus_model}, we define the linear consensus dynamics. \Cref{sec:attack_trust_model} describes the attack model for malicious agents and introduces the stochastic trust observations. In \Cref{sec:detection}, we present the detection algorithm used by legitimate agents to classify their neighbors as either legitimate or malicious, and describe how agents assign consensus weights based on these classifications.
\subsection{Notation}
We denote the absolute value of a scalar and the cardinality of a finite set by $|\!\cdot\!|$. We write $x_{i}$ and $A_{ij}$ for the $i$-th entry of a vector $x$ and the $ij$-th entry of a matrix $A$, respectively. When the notation is heavy, we use $[\cdot]_i$ and $[\cdot]_{ij}$, respectively, for the $i$-th entry of a vector and the $ij$-th entry of a matrix, such as when a vector or a matrix are expressed as products and/or summations of other vectors and matrices. For matrices $A$ and $B$, we write $A>B$ (or $A\ge B$) when $A_{ij}>B_{ij}$ 
(or $A_{ij} \ge B_{ij}$) for all $i,j$. 
The backward matrix product of the matrices $H(k)$, is defined as:
\begin{align}
    \prod_{k=\tau}^t H (k):= \begin{cases}
 H(t) \cdots H (\tau+1)  H (\tau) \: &\text{if } t \ge \tau, \\
I \: &\text{otherwise},
\end{cases}
\end{align}
where $I$ corresponds to the identity matrix. 
\subsection{Consensus in Presence of Untrustworthy Agents}\label{sec:consensus_model}
We consider the consensus process among multiple agents defined by the set $\mathcal{N}= \{1,\ldots, N\}$. The agents exchange information with their neighbors through a static undirected graph  $G(\mathcal{N}, \mathcal{E})$, where $\mathcal{E}\subseteq \mathcal{N} \times \mathcal{N}$ is the set of undirected edges among the agents. Thus, we have $(i,j) \in \ccalE$ if and only if $(j,i) \in \ccalE$.
Each agent $i \in \ccalN$ has a set of neighboring agents
denoted by $\Ni=\{ j \in \ccalN\mid (i,j) \in \ccalE\}$.
The set of agents $\mathcal{N}$ consists of two disjoint subsets: legitimate agents, who are always trustworthy, and malicious agents, who may or may not be trustworthy. The set of legitimate and malicious agents are denoted, respectively, by $\mathcal{L}$ and $\mathcal{M}$, satisfying $\mathcal{L}\cup\mathcal{M}=\mathcal{N}$ and $\mathcal{L}\cap \mathcal{M}=\emptyset$. These sets are fixed over time and assumed to be unknown by legitimate agents. For each legitimate agent $i\in \mathcal{L}$, we denote its set of legitimate neighbors by $\Nil = \Ni \cap \mathcal{L}$ and malicious neighbors by 
$\Nim = \Ni \cap \mathcal{M}$. Legitimate agents assign nonnegative weights $w_{ij}(t)$ to their neighbors, with $w_{ij}(t) \in [0,1]$ if $(i,j) \in \mathcal{E}$ and $w_{ij}(t) = 0$ otherwise. These weights change over time and we will detail how agents determine these weights later on. All legitimate agents $i\in \ccalL$ start with an arbitrary initial value $x_i(0)\in\reals$ and update their values according to the following consensus dynamic, starting at some time $T_0\ge0$, $\forall t \ge T_0-1$,
\begin{equation} \label{eq_con}
     x_i(t+1)=w_{ii}(t)x_i(t)+\sum_{j \in \mathcal{N}_i} w_{ij}(t)x_j(t),
 \end{equation}
 where $x_i(t) \in \reals$ for all $i \in \ccalN$. 
 Before the start time $T_0$, the legitimate agents do not update their values, i.e., $x_i(t)=x_i(0)$ for all $0\le t < T_0$. According to Eq. \eqref{eq_con}, each legitimate agent $i \in \ccalL$ updates its value as a weighted average of its own and its neighbors’ values, with $w_{ii}(t) > 0$, $w_{ij}(t) \ge 0$, and $w_{ii}(t) + \sum_{j \in \mathcal{N}i} w_{ij}(t) = 1$.

 We consider the cases where agents' initial values lie within the interval $[-\eta,\eta]$ for some $\eta>0$ that is known to all agents. Since legitimate agents update their values in the consensus process by taking a convex combination of their own value and those of their neighbors, this assumption ensures that $|x_i(t)| \le \eta$ for all $i\in\mathcal{N}$ and $t\geq0$.
 Malicious agents also send any values within $[-\eta, \eta]$ but avoid values outside this interval, as these would result in their immediate detection. We detail malicious agents behavior in the next section.


Next, we express the consensus dynamics in matrix form for use in later analysis. We let the vector $x(t) \in \reals^N$ consist of the agents' values, where $x_i(t)$ is the value of agent $i$. Given the disjoint sets $\LL$ and $\M$ of legitimate and malicious agents, without loss of generality, we assume that the agents are indexed in a such way that the last $\M$ agents are malicious. Thus, 
we can write $x(t)=[ \xl(t)^\top, \xm(t)^\top]^\top$ without loss of generality, where $\xl(t)\in \mathbb{R}^{|\ccalL|}$ represents the values of legitimate agents and $\xm(t) \in \mathbb{R}^{ |\ccalM|}$ those of malicious agents. Then, the consensus process~\eqref{eq_con} in the vector form is given by
 \begin{equation} \label{eq_con_dy}
\xl(t+1) 
=
     \begin{bmatrix}
\wl(t) & \wm(t)
\end{bmatrix}
\cdot
  \begin{bmatrix}
\xl(t) \\
\xm(t) 
\end{bmatrix},
\end{equation}
where $\wl(t) \in \mathbb{R}^{|\ccalL|\times |\ccalL|}$ and $\wm(t) \in \mathbb{R}^{|\ccalL|\times |\ccalM|}$ are the weight matrices that legitimate agents associate with legitimate and malicious agents, respectively.
In what follows, the weight matrices $\wl(t)$ and $\wm(t)$ will depend on the start time $T_0$.
To capture this dependence, we will write $\xl(T_0,t)$ instead of $\xl(t)$.
For all $t\geq T_0$, we decompose $\xl(T_0, t)$ into two terms to separate the contributions of legitimate and malicious agents as follows: for all $t \ge T_0-1$,
\begin{equation}\label{eq_con_sum}
     \xl(T_0,t)= \Tilde{x}_{\ccalL}(T_0,t)+\phi_{\ccalM}(T_0,t),
 \end{equation}
where 
\begin{align}  
    \Tilde{x}_L(T_0,t)&= \bigg ( \prod_{k=T_0-1}^{t-1} \wl(k) \bigg)\xl(0), \label{eq_con_part1}\\
    \phi_{\ccalM}(T_0,t)&=  \sum_{k=T_0-1}^{t-1} \bigg ( \prod_{s=k+1}^{t-1} \wl(s) \bigg) \wm(k) \xm(k). \label{eq_sep_Dy}
\end{align}
The term $\Tilde{x}_L(T_0,t) \in  \mathbb{R}^{|\ccalL|}$ results from consensus dynamics among legitimate agents, while the vector $\phi_{\ccalM}(T_0,t) \in \mathbb{R}^{|\ccalL|}$ captures the influence of malicious agent inputs $\xm(k) \in \mathbb{R}^{|\ccalM|}$. The relations~\eqref{eq_con_sum}-\eqref{eq_sep_Dy} are central in the subsequent analysis, as they capture the consensus dynamics among the legitimate agents in terms of the starting time $T_0$, the initial vector $x(0)$, and the influence of the malicious inputs.

\subsection{Attack and Trust Models}\label{sec:attack_trust_model}
We consider an attack model where at each time step, every malicious agent decides whether to attack the system. For a malicious agent $m\in \M$, we denote its attack decision at time $t$ by the indicator random variable $f_m(t) \in \{0,1\}$, where $\{f_m(t)=1\}$ indicated the event of an attack. When attacking, a malicious agent is allowed to transmit any value within the interval $[-\eta,\eta]$. 

We focus on the setting where each legitimate agent $i$ gathers a stochastic trust observation $\alpha_{ij}(t)\in [0,1]$ associated with transmissions from its neighbor $j$, where a larger value corresponds to a higher likelihood of an attack originating from neighbor $j$ at time $t$. This stochastic inter-agent trust model captures scenarios where legitimate agents can leverage physical channels of information, such as sensor observations and wireless fingerprints, to assess the trustworthiness of their neighbors, and studied in various other works \cite{gil2017guaranteeing, yemini2021characterizing,yemini2022resilentopt,akgun2023learning,cavorsi2023ICRA,ballotta2024role,yong2024}. The side information from these physical sources enables cross-validation of transmissions through the physical environment.  See ~\cite{gil2017guaranteeing, cavorsi2023ICRA} for examples of such trust observations and how they can be computed.

We have the following assumption on the connectivity of legitimate agents and the trust observations $\alpha_{ij}(t)$ for transmissions among the legitimate agents.
\begin{assumption}\label{leg_ag}[Legitimate Agents]
Assume that:\\
\noindent (1)
The subgraph $G_{\ccalL}=(\ccalL, \ccalE_{\ccalL})$ induced by the legitimate agents $\mathcal{L}$ is connected, where $\ccalE_{\ccalL}= \{(i,j) \in \ccalE\mid i,j \in \ccalL\}$.\\
\noindent (2) For any legitimate agent $i\in \LL$ and any of its legitimate neighbors $j$, the trust observations $\alpha_{ij}(t)$ are independent over time. 
Moreover, these observations
have static expectations that are uniform across the legitimate agents, i.e., for all $t\ge0$,
\[\mathbb{E}(\alpha_{il}(t))=E_\LL\qquad\hbox{for all $i\in \LL$ and $\ell\in \Nil$}.\]
\end{assumption}

The malicious agents can choose to attack or not with some time-varying probability. Moreover, the probability of an attack at any time $t\ge1$ can depend on the past outcomes for all $t\ge 1$, i.e.,
$\mathbb{P}(f_m(t)=1\mid f_m(0),\ldots,f_m(t-1))$.
To formalize this, for any $m\in \M$, we define the history of agent $m$'s attack decisions up to time $t$ as
\begin{equation}\label{eq-hist}
\mathcal{F}_m(t)=\{f_m(0),\ldots, f_m(t)\}\qquad\hbox{for all $t\ge0$},
\end{equation}
where $\mathcal{F}_m(-1)=\emptyset$. For all $t\ge1$, 
let $p_m(t)$ be the {\it smallest conditional probability} of the events $\{f_m(t)=1\mid \mathcal{F}_m(t-1)\}$ for all  possible past outcomes $\mathcal{F}_m(t-1)$, i.e., 
for all $m\in\M$ and $t\ge1$,
\begin{equation}\label{def-pmt}
p_m(t)\!=\!\min_{\atop \mathcal{F}_m(t-1)\in \{0,1\}^{t}}\!
\mathbb{P}(f_m(t)=1\mid \mathcal{F}_m(t-1)).
\end{equation}
Also, let
\[p_m(0)=\mathbb{P}(f_m(0)=1)\qquad\hbox{for all }m\in\M.\]
We use the following assumption for the malicious agents.
\begin{assumption}
\label{mal_ag}[Malicious Agents]
 Assume that:\\ 
\noindent (1)
The conditional expectations of the trust observations received by a legitimate agent $i$ from a malicious neighbor $m$ satisfy the following for all $t\ge0$, $i\in \LL$ and $m\in \Nim$:
\[\mathbb{E}(\alpha_{im}(t)\mid f_m(t)=0)
=E_\LL,\]
\[\mathbb{E}(\alpha_{im}(t)\mid f_m(t)=1)
=\mu_m(t).\]
We also have $E_\LL>E_\M$ where $E_\M=\max_{m\in\M, t\geq0}\mu_m(t)$.\\
\noindent (2)
For any legitimate agent $i\in\LL$ and any of its malicious neighbors $m\in \Nim$, given $f_m(t)$, the conditional trust observations  $\alpha_{im}(t)|f_m(t)$ are independent over time $t$.
\end{assumption}

Assumption~\ref{leg_ag}(1) ensures that the subgraph induced by legitimate agents is connected. This is a standard and relatively mild requirement in the resilient consensus literature~\cite{yemini2021characterizing, ballotta2024role}, and it is weaker than the strong robustness conditions often required in deterministic settings without stochastic trust observations~\cite{ishii2022overview}. Assumption~\ref{leg_ag}(2) posits independence of trust observations over time and uniform expected trust values across legitimate agents. This aligns with existing works that incorporate stochastic trust observations~\cite{yemini2021characterizing, yemini2022resilentopt, akgun2023learning, ballotta2024role, yong2024}.

In contrast, our assumptions for malicious agents, particularly in Assumption~\ref{mal_ag}, are more general than those in prior work, including our previous conference paper~\cite{aydin2024multi}. Many existing studies assume that malicious agents attack at every time step, i.e., $p_m(t)=1$ for all $t$ \cite{yemini2021characterizing,yemini2022resilentopt,akgun2023learning,ballotta2024role,yong2024}. They also typically assume that trust observations are independent over time, identically distributed for each pair of agents, and stationary in expectation \cite{yemini2021characterizing,yemini2022resilentopt,akgun2023learning,ballotta2024role,yong2024,aydin2024multi}. Our model relaxes these assumptions in two important ways. First, we allow malicious agents to vary their attack probabilities $\mathbb{P}(f_m(t)=1)$ over time based on their own histories and potentially in coordination with other malicious agents. This introduces potential temporal dependencies into the sequence of trust observations. Second, we allow the expected trust values received from malicious agents to vary across agents and over time, provided they remain uniformly bounded above by $E_\M$. As a result, our analysis does not rely on the strong independence or stationarity assumptions common in prior works. Instead, we accommodate adaptive and time-correlated trust observations while still guaranteeing detection and consensus under appropriate conditions.

In our framework, malicious agents influence the consensus process in two distinct ways: 1) By controlling their attack probabilities $\mathbb{P}(f_m(t)=1)$, and 2) By choosing the values $x_m(t)\in[-\eta,\eta]$ they transmit. This flexibility allows for a broad range of malicious behaviors, including collaborative or strategic attacks where adversaries may coordinate both when to attack and what values to transmit in order to maximize disruption. This stands in contrast to previous models that typically only allow malicious agents to choose transmitted values arbitrarily while assuming fixed or non-adaptive attack schedules \cite{yemini2021characterizing, ballotta2024role, aydin2024multi}. Importantly, we do not impose any assumptions on the behavior of malicious agents when they are not attacking, i.e., how they choose $x_m(t)$ when $f_m(t)=0$. Modeling this non-attacking behavior would require additional structure, such as assuming that non-attacking agents follow the same consensus dynamics as legitimate agents. To maintain generality and accommodate a wide range of adversarial strategies, we avoid such assumptions. As a result, our consensus analysis is more conservative and yields worst-case upper bounds on the influence of malicious agents. Finally, our attack model assumes that malicious agents broadcast the same value $x_m(t)$ to all their neighbors, and make a single attack decision $f_m(t)$ per time step. However, our detection and convergence analysis extends to more general models where a malicious agent makes neighbor-specific decisions $f_{mj}(t)$ and transmit different values $x_{mj}(t)$ to each neighbor.

\subsection{Trusted Neighborhood Learning}\label{sec:detection}
In this section, we present an algorithm that legitimate agents use to identify their malicious neighbors. To identify their malicious neighbors, at time $t$, the legitimate agents use the history of the trust observations $\{\alpha_{ij}(k)\}_{k=0}^t$ to select their trustworthy neighbors. This selection is done 
by assigning positive weights $w_{ij}(t)>0$ to such neighbors in the consensus process~\eqref{eq_con}. In the algorithm, every legitimate agent $i$ uses the aggregate trust observations about its neighbor $j$, defined as: for all $t\ge0$,
\begin{equation}\label{eq-agt}
    \bij(t)=\sum_{k=0}^{t} \aij(k)\qquad  \forall \: i \in \ccalL\hbox{ and } j \in \Ni. 
\end{equation}

In the trusted neighborhood learning algorithm (\Cref{alg_trust}), every legitimate agent $i \in \mathcal{L}$ selects its trusted neighbors based on the aggregate trust values $\beta_{ij}(t)$, as follows. At first,  each legitimate agent $i$ identifies its most trusted neighbor $\bar{j}$ (that could be malicious or legitimate at a given time), with the largest aggregate trust value in its neighbor set, i.e., $\beta_{i\bar{j}}(t)=\max_{j\in \Ni}\beta_{ij}(t)$. At second, agent $i$ evaluates whether the trust values of its other neighbors are sufficiently close to this most trusted value, using a time-varying threshold $\xi_t$ 
on the difference $\beta_{i\bar{j}}(t) - \beta_{ij}(t)$. The algorithm outputs the trusted neighbor set $\hat{\ccalN}_i(t)$ for every $i\in\LL$. As seen from the algorithm, all neighbors $j\in\Ni$ that attain the maximum $\max_{j\in\Ni}\beta_{ij}(t)$ are always included in the set $\hat{\ccalN}_i(t)$ of trusted neighbors. 

This algorithm is based on two key observations. First, by Assumption~\ref{leg_ag}, every legitimate agent has at least one legitimate neighbor, and the trust observations from legitimate neighbors are independent over time with identical expectations. Therefore, the difference between the aggregate trust values of two legitimate neighbors (e.g., $\beta_{il_1}(t) - \beta_{il_2}(t)$) is expected to remain small relative to time $t$. Second, the trust values form malicious neighbors have a strictly lower expected value when they are attacking. Therefore, as long as they attack frequently enough in probability (as formally characterized in \Cref{sec:convergence_of_consensus}), the gap between the aggregate trust values of legitimate and malicious neighbors grows sufficiently large over time. Consequently, a suitably chosen detection threshold $\xi_t$ can distinguish between legitimate and malicious neighbors. The design of the threshold sequence $\xi_t$ and its relationship to the frequency of attacks (captured by the lower bound on the attack probability $p_m(t)$) play an important role in ensuring the correctness of the algorithm. In our prior work~\cite{aydin2024multi}, we proposed Algorithm~\ref{alg_trust} using a specific threshold of the form $\xi_t=\xi(t+1)^\gamma$ where $\xi>0$ and $\gamma\in (0.5,1)$. In this work, we extend the analysis and show that the algorithm remains effective for a broader class of detection thresholds. In \Cref{sec:convergence_of_consensus}, we derive sufficient conditions under which a general sequence $\xi_t$ ensures almost sure convergence of the consensus algorithm. In \Cref{sec:deviation}, we analyze the impact of a specific choice of $\xi_t$ on the deviation from the nominal consensus value in the absence of malicious agents. 

\begin{algorithm}[H] 
   \caption{Trusted Neighborhood Learning (for every $i\in\LL$)}
\label{alg_trust}
\begin{algorithmic}[1]\label{alg_trust_neig}
  \STATE {\bfseries Input:} Time-varying threshold $\xi_t>0$.
\STATE Agent $i \in \ccalL$ selects one of its most trusted agent $\bar{j}(t)\in {\rm Argmax}_{j \in \mathcal{N}_i} \beta_{ij}(t)$.
\STATE  Agent $i \in \ccalL$ checks if $ \beta_{i\bar{j}(t)}(t)- \beta_{ij}(t) \le \xi_t$ $ \forall j \in \Ni$.
\STATE Agent $i \in \ccalL$ forms its trusted neighborhood $\hat{\ccalN}_i(t)=\{j \in \Ni\mid \beta_{i\bar{j}(t)}(t)- \beta_{ij}(t) \le \xi_t \}$.
 \STATE {\bfseries Output:} The set $\hat{\ccalN}_i(t)$ of trusted neighborhood.
   \end{algorithmic}
\end{algorithm}
Upon executing Algorithm~\ref{alg_trust},
the legitimate agents use their trusted neighborhoods $\hat{\ccalN}_i(t)$ to define the weights $w_{ij}(t)$
for the consensus process~\eqref{eq_con}, as follows:
\begin{equation} \label{eq_wij}
    w_{ij}(t)= \begin{cases}
 \frac{1}{n_{w_i}(t)} \: &\text{if } j \in \hat{\ccalN}_i(t), \\
1-\sum_{\ell \in \hat{\ccalN}_i(t)} w_{i\ell}(t) \: &\text{if } j=i,\\
0 \: &\text{otherwise},
\end{cases}
\end{equation}
where $n_{w_i}(t)= \max\{|\hat{\ccalN}_i(t)|+1, \kappa \} \ge 1$ 
and $\kappa >0$ is a common parameter bounding the effect of other agents on the consensus process.
We let $W(t)$ be the matrix with entries $w_{ij}(t)$ as defined in~\eqref{eq_wij}.
We also define the {\it nominal matrix} $\widebar{W}_{\ccalL}$ as the weight matrix
that would have been formed according to~\eqref{eq_wij} if the legitimate agents have classified their neighbors correctly, i.e.,
for all the legitimate agents $i\in \ccalL$,
\begin{equation} \label{eq_barwij}
[\widebar{W}_{\ccalL}]_{ij}= \begin{cases}
\frac{1}{\max\{ |\Nil|+1, \kappa\}} \: &\text{if } j \in \Nil,\\
1-\frac{\left|\Nil \right|}{\max\{ |\Nil|+1, \kappa\}} \: &\text{if } j=i,\\
0 \: &\text{otherwise}.
\end{cases}
\end{equation}
The nominal matrix $\widebar{W}_{\ccalL}$ is the matrix that the legitimate agents would have used in the absence of malicious agents. It serves as a reference for evaluating the performance of the consensus algorithm, as it captures the ideal, unperturbed case without adversarial influence. In the next section, we analyze the performance of both the trusted neighborhood learning algorithm and the resulting consensus dynamics.

\section{Analysis} \label{sec::conv}

In this section, we analyze the convergence properties of the consensus dynamics in~\eqref{eq_con_dy}. The performance of the consensus algorithm critically depends on the trusted neighborhood learning algorithm: for the system to behave similarly to the ideal case, legitimate agents must correctly identify their legitimate neighbors and exclude malicious ones. To this end, we first analyze the performance of the detection algorithm in \Cref{sec:prelim} and \Cref{sec:detection_analysis}. Our analysis is based on the existence of a (random but finite) time after which each legitimate agent correctly classifies all of its neighbors almost surely. In \Cref{sec:convergence_of_consensus}, we establish sufficient conditions for the existence of such a time and show that, under these conditions, the consensus algorithm converges almost surely. Next, in \Cref{sec:char_tf}, we characterize how quickly agents reach this correct classification time in probability. In \Cref{sec:deviation}, we analyze the deviation from the nominal consensus value---the consensus that would have been achieved in the absence of malicious agents. Finally, in \Cref{sec:conv_rate}, we investigate the rate of convergence.

\subsection{Preliminary Results}\label{sec:prelim}
In this section, we provide some preliminary results
that will be used to analyze the performance of the trusted-neighborhood learning algorithm (\Cref{alg_trust_neig}).
Our analysis leverages a concentration inequality to bound the probability of misclassifying a legitimate neighbor as malicious and vice versa. We first present Hoeffding's Lemma.

\begin{lemma} [Hoeffding's Lemma (\cite{concentration-ineq2013}, Lemma 2.2, pg.\ 27]\label{lem_hoeff}
    Let $X$ be a real random variable taking values in the interval $[a,b]$ almost surely. Then, for any $ \l>0$, it holds
    \begin{equation*}
        \mathbb{E}(e^{\l X}) \le e^{\l\mathbb{E}(X)+ \l^2 (b-a)^2/8}.
    \end{equation*}
\end{lemma}

\begin{figure}
     \centering
     \begin{subfigure}[b]{0.4\textwidth}
         \centering
         \includegraphics[width=\textwidth]{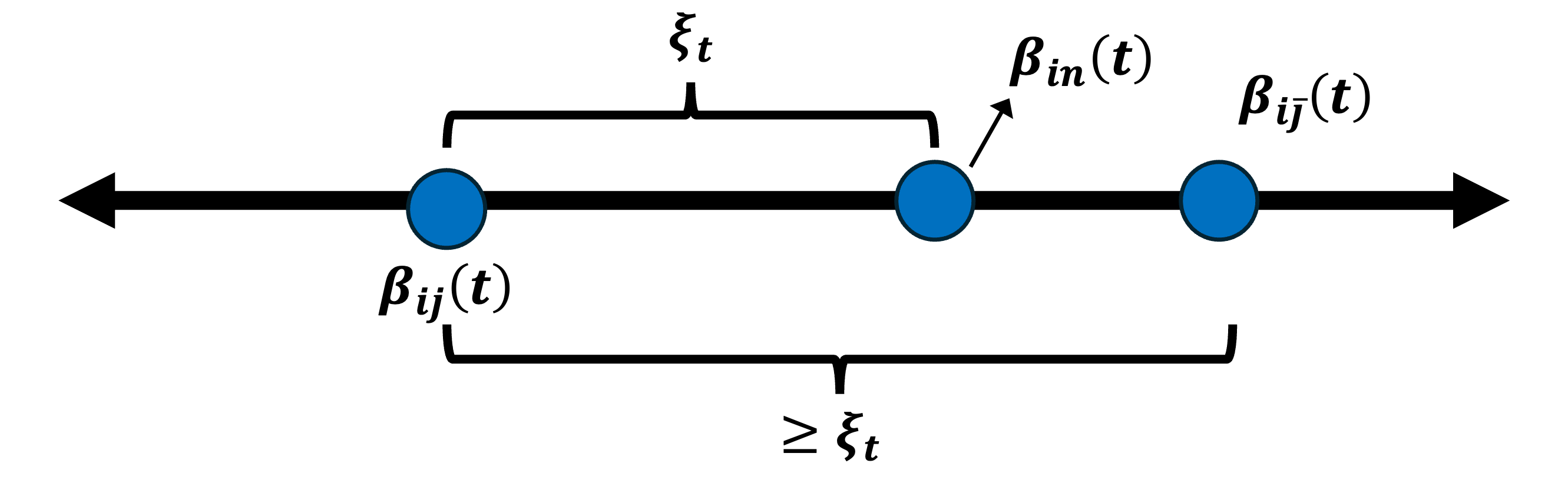}
         \caption{Legitimate neighbor $j$ is misclassified as malicious}
         \label{fig:alg_intuition_a}
     \end{subfigure}
     \hfill
     \begin{subfigure}[b]{0.4\textwidth}
         \centering
         \includegraphics[width=\textwidth]{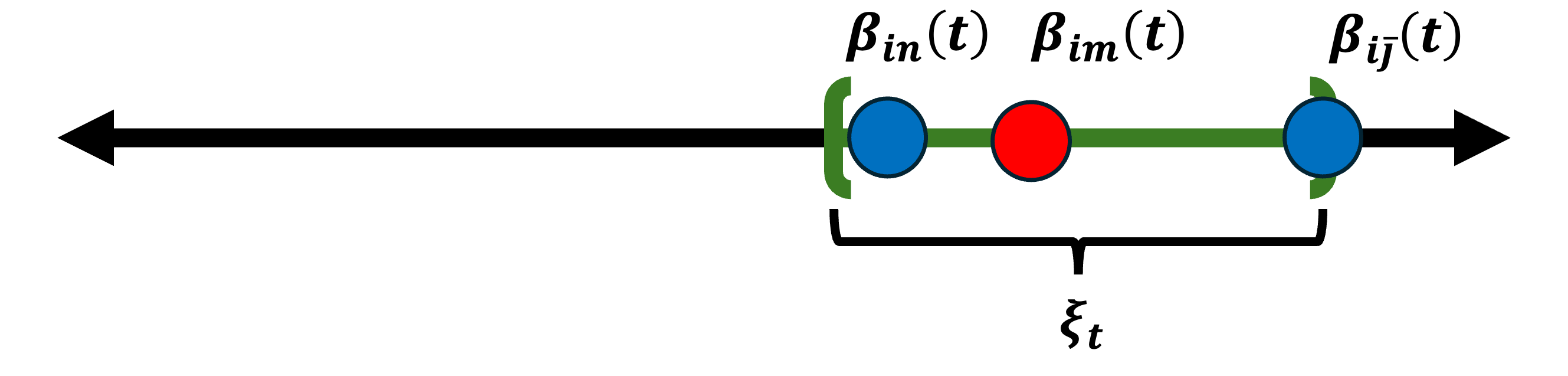}
         \caption{Malicious neighbor $m$ is misclassified as legitimate}
         \label{fig:alg_intuition_b}
     \end{subfigure}
    \caption {\footnotesize Trusted neighborhood learning algorithm for a legitimate agent $i\in\mathcal{L}$ is illustrated. Legitimate neighbors are shown in blue and malicious neighbors in red. Aggregate trust values are placed on a number line with larger values to the right. The green bracketed region represents the trusted region $\xi_t$ from Algorithm~\ref{alg_trust_neig}. (a) The accumulated trust value $\beta_{ij}(t)$ falls at least $\xi_t$ to the left of $\beta_{in}(t)$, implying that its gap to $\beta_{i\bar{j}}(t)$ is also at least $\xi_t$. (b) Since $\beta_{im}(t)$ lies within $\xi_t$ distance of $\beta_{i\bar{j}}(t)$, other aggregate trust values can either lay on its left or remain within $\xi_t$ distance on its right.}
    \label{fig:alg_intuition}
\end{figure}

Using Hoeffding's Lemma, for a legitimate agent $i$ and an arbitrary scalar $r$, we derive upper bounds on the probabilities $\mathbb{P}(\beta_{im}(t)-\b_{i\ell}(t) > r)$ and $\mathbb{P}(\beta_{ij}(t)-\b_{i\ell}(t) > r)$ for a malicious neighbor $m$, and legitimate neighbors $j$ and $\ell$.

\begin{lemma}\label{lem-gen-bound}
Let Assumptions~\ref{leg_ag} and~\ref{mal_ag} hold.
Let $r\in\mathbb{R}$ and $\lambda>0$ be arbitrary. Then, the following statements hold for all legitimate agents $i\in\LL$  and all $t\ge0$:
\begin{itemize}
\item[(a)]
For all legitimate neighbors $j,\ell\in\Nil$ of agent $i$, 
we have
\[\mathbb{P}\left(\beta_{ij}(t)-\b_{i\ell}(t)>  r\right)\le e^{\l^2(t+1)/2-\l r}.\]
\item[(b)]  
For all malicious neighbors $m\in \Nim$ and all legitimate neighbors $\ell\in\Nil$ of agent $i$, we have
\begin{align*}
&\mathbb{P}\left(\beta_{im}(t)-\b_{i\ell}(t)>  r \right) \cr
&\qquad \le 
e^{{\lambda (E_\M-E_\LL)} \sum_{k=0}^t p_m(k)+ \l^2(t+1)/2-\l r}.
\end{align*}
\end{itemize}
\end{lemma}
\begin{proof}
Let $r\in\mathbb{R}$ and $\lambda>0$ be arbitrary.
Consider a legitimate agent $i$ and two of its neighbors $j,j'\in\Ni$. 
Then, for the difference $\beta_{ij}(t)-\b_{ij'}(t)$ of the aggregated trust values, we have for all $t\ge0$,
\begin{align}\label{eq-mas}
    \mathbb{P}\left(\beta_{ij}(t)-\b_{ij'}(t)>  r \right)
    &=\mathbb{P}\left(e^{\l (\beta_{ij}(t)-\b_{ij'}(t))}> e^{\l r}\right)\cr
    &\le e^{-\l r} \mathbb{E}(e^{\l (\beta_{ij}(t)-\b_{ij'}(t))}).
    \end{align}
where the inequality follows from Markov's Inequality. 
We now consider the parts (a) and (b) separately. \\
\noindent
(a) \ When both agents $j$ and $j'$ are legitimate neighbors of agent $i$, 
by the definition of the aggregate trust observations $\b_{ij}(t)$ and using 
the independence of the trust observations $\a_{ij}(k)$ over time (Assumption~\ref{leg_ag}(2)), from relation~\eqref{eq-mas} we obtain
\begin{align*}
    \mathbb{P}\left(\beta_{ij}(t)-\b_{ij'}(t)>  r \right)
    &\le e^{-\l r} \prod_{k=0}^t\mathbb{E}(e^{\l (\a_{ij}(k)-\a_{ij'}(k))}).
    \end{align*}
Since $\a_{ij}(t)\in[0,1]$ for all $i,j$ and all $t\ge0$,
we have that $\a_{ij}(k)-\a_{ij'}(k)\in[-1, 1]$ for all $k\ge0$.
Applying Hoeffding's Lemma (Lemma~\ref{lem_hoeff}) to the variables 
$\l (\a_{ij}(k)-\a_{ij'}(k))$, we obtain for all $k\ge0$,
\begin{equation*}
    \mathbb{E}(e^{\l (\a_{ij}(k)-\a_{ij'}(k))})
\le e^{\l \mathbb{E}(\a_{ij}(k)-\a_{ij'}(k))+\l^2/2}
=e^{\l^2/2},
\end{equation*}
where the last equality follows from 
the assumption that all trust observations $\a_{ij}(k)$ have the same expected value $E_\LL$ for all legitimate neighbors of $i$ (Assumption~\ref{leg_ag}(2)).
Combining the preceding two relations, we find that for all $t\ge0$,
\[\mathbb{P}\left(\beta_{ij}(t)-\b_{ij'}(t)>  r \right)
\le e^{\l^2(t+1)/2-\l r},\]
thus showing the relation in part (a).\\
\noindent
(b) \ When neighbor $j$ is malicious and $j'$ is legitimate, i.e., $j=m$
for some $m\in\Nim$ and $j'=\ell$ for some $\ell\in\Nil$, from relation~\eqref{eq-mas} we have for all $t\ge0$,
\begin{align}\label{eq-jed}
&\mathbb{E}\left(e^{\l (\beta_{im}(t)-\b_{i\ell}(t))}\right)\cr
&\!=\!\!\sum_{\mathcal{F}_m(t-1)}\!\!\!\mathbb{E}\left(e^{\l (\beta_{im}(t)-\b_{i\ell}(t))}\mid\mathcal{F}_m(t-1)\right)\mathbb{P}(\mathcal{F}_m(t-1))\cr
&= \!\!\sum_{\mathcal{F}_m(t-1)}\!\!\!\mathbb{E}\left(\prod_{k=0}^{t-1} e^{\l (\a_{im}(k)-\a_{i\ell}(k))}\mid\mathcal{F}_m(t-1)\!\!\right)\mathbb{P}(\mathcal{F}_m(t-1))\cr
&\qquad \times \mathbb{E}(e^{\lambda (\a_{im}(t)-\a_{i\ell}(t))}\mid \mathcal{F}_m(t-1)),
\end{align}
where the last equality is obtained by using the  definition of accumulated trust observations $\b_{ij}(k)$, and by using Assumption~\ref{mal_ag}(2)  on the independence of the trust observations. 
Since $\a_{ij}(t)\in[0,1]$ for all $i,j$ and all $t\ge0$,
we have that $-1\le \a_{im}(t)-\a_{i\ell}(t)\le 1$ for all $t\ge0$.
Thus, by applying Hoeffding's Lemma (Lemma~\ref{lem_hoeff}, for the last term in relation~\eqref{eq-jed}
we obtain
\begin{align}\label{eq-dva}
&\mathbb{E}(e^{\lambda (\a_{im}(t)-\a_{i\ell}(t))}\mid \mathcal{F}_m(t-1)) \cr
&\quad
\le e^{\l\mathbb{E}(e^{\lambda (\a_{im}(t)-\a_{i\ell}(t))}\mid \mathcal{F}_m(t-1))+ \l^2/2}).
\end{align}
Using the iterated expectation rule, we have
\begin{align*}
&\mathbb{E}(\l(\a_{im}(t)-\a_{i\ell}(t))\mid \mathcal{F}_m(t-1))
\cr
&=
\mathbb{E}(\lambda (\a_{im}(t)-\a_{i\ell}(t))\mid \mathcal{F}_m(t-1),f_m(t)=0)\cr
&\quad \times \mathbb{P}(f_m(t)=0\mid \mathcal{F}_m(t-1))\cr
&\quad + \mathbb{E}(\lambda (\a_{im}(t)-\a_{i\ell}(t))\mid \mathcal{F}_m(t-1),f_m(t)=1)\cr
&\quad \times \mathbb{P}(f_m(t)=1\mid \mathcal{F}_m(t-1)).
\end{align*}
Using the assumptions on the expected trust observations (Assumption~\ref{leg_ag}(2) and Assumption~\ref{mal_ag}(1)), we obtain
\begin{align*}
&\mathbb{E}(\lambda (\a_{im}(t)-\a_{i\ell}(t))\mid \mathcal{F}_m(t-1))
\cr
&=0+\lambda (\mu_m(t)-E_\LL) \mathbb{P}(f_m(t)=1\mid \mathcal{F}_m(t-1))\cr
&\le \lambda (E_\M-E_\LL) \mathbb{P}(f_m(t)=1\mid \mathcal{F}_m(t-1))
\end{align*}
where $E_\M=\max_{m\in \M, t\geq0}\mu_m(t)$ and $E_\M < E_\LL$ (Assumption~\ref{mal_ag}(1)). Since $E_\M< E_\LL$, and by the definition of $p_m(t)$ in \Cref{def-pmt} as the smallest conditional probability, we have 
\begin{align*}
\mathbb{E}(\lambda (\a_{im}(t)-\a_{i\ell}(t))\mid \mathcal{F}_m(t-1))
\le \lambda (E_\M-E_\LL) p_m(t).
\end{align*}
By combining the preceding inequality with~\eqref{eq-dva},
we obtain for all $t\ge0$,
\begin{align*}
\mathbb{E}(e^{\lambda (\a_{im}(t)-\a_{i\ell}(t))}\!\mid\! \mathcal{F}_m(t-1))
\!\le\! e^{{\lambda (E_\M-E_\LL)} p_m(t)+ \l^2/2}.
\end{align*}
Upon substituting the preceding inequality back in relation~\eqref{eq-jed} and using the fact that $p_m(t)$ does not depend on $\mathcal{F}_m(t-1)$,
we obtain for all $t\ge0,$
\begin{align}\label{eq-cet}
&\mathbb{E}\left(e^{\l (\beta_{im}(t)-\b_{i\ell}(t))}\right)\cr
&
\le \mathbb{E}\left(\prod_{k=0}^{t-1} e^{\l (\a_{im}(k)-\a_{i\ell}(k))}\right)
e^{{\lambda (E_\M-E_\LL)} p_m(t)+ \l^2/2}. \ \ 
\end{align}
By repeating the process iteratively, i.e., writing the expectation in~\eqref{eq-cet} in terms of the conditional expectation on $\mathcal{F}_m(t-2)$ and so on, we obtain
\begin{align*}
\mathbb{E}\left(e^{\l (\beta_{im}(t)-\b_{i\ell}(t))}\right)
\le 
e^{{\lambda (E_\M-E_\LL)} \sum_{k=0}^t p_m(k)+ \l^2(t+1)/2}.
\end{align*}
By combining the preceding relation with relation~\eqref{eq-mas},
we obtain the relation in part (b). 
\end{proof}

In the next result, we refine Lemma~\ref{lem-gen-bound} by optimizing the choice of the parameter $\l$.
\begin{lemma}\label{lem-opt-bound}
Let Assumptions~\ref{leg_ag} and~\ref{mal_ag} hold.
Then, the following statements are valid for all legitimate agents $i\in\LL$:
\begin{itemize}
\item[(a)] 
Let $r>0$. Then, for any legitimate $\ell\in\Nil$ and any other neighbor $j\in\Ni$, $j\ne\ell$, we have for all $t\ge0$,
\[\mathbb{P}\left(\beta_{ij}(t)-\b_{i\ell}(t)> r\right)\le e^{-r^2(t+1)^{-1}/2}.\]
\item[(b)]  Let $r<0$. Then, for all malicious neighbors $m\in \Nim$ and all legitimate neighbors $\ell\in\Nil$ of agent $i$, at time $t\ge0$ such that $(E_\LL-E_\M) \sum_{k=0}^t p_m(k)+r>0$, we have
\begin{align*}
&\mathbb{P}\left(\beta_{im}(t)-\b_{i\ell}(t)>  r \right) \cr
&\qquad \le 
e^{-\left((E_\LL-E_\M) \sum_{k=0}^t p_m(k)+r\right)^2(t+1)^{-1}/2}.
\end{align*}
\end{itemize}
\end{lemma}
\begin{proof}
(a) \ We consider separately the cases when $j$ is legitimate and when it is malicious. Suppose $j\in\Nil$. Then, by Lemma~\ref{lem-gen-bound}(a), for all legitimate neighbors $j,\ell\in\Nil$ of agent $i$, 
we have for all $t\ge0$,
\[\mathbb{P}\left(\beta_{ij}(t)-\b_{i\ell}(t)>  r\right)\le e^{\l^2(t+1)/2-\l r}.\]
Taking the minimum over $\l>0$ on the right hand side of the preceding relation, we can see that the minimum is attained at $\l^*=r(t+1)^{-1}$,
which when substituted in the preceding relation yields the stated inequality.

Suppose now that $j$ is malicious neighbor, i.e., $j=m$
with $m\in \Nim$. Then, by Lemma~\ref{lem-gen-bound}(b), for all malicious neighbors $m\in \Nim$ and all legitimate neighbors $\ell\in\Nil$ of agent $i$, we have for all $t\ge0$,
\begin{align}\label{eq-des}
&\mathbb{P}\left(\beta_{im}(t)-\b_{i\ell}(t)>  r \right) \cr
&\qquad \le 
e^{{\lambda (E_\M-E_\LL)} \sum_{k=0}^t p_m(k)+ \l^2(t+1)^2/2-\l r}.
\end{align}
The minimum value of the right hand side of the preceding relation, over $\l>0$, is attained at $\l^*=(r+(E_\LL-E_\M) \sum_{k=0}^t p_m(k))(t+1)^{-1}$,
which when substituted in the preceding relation yields 
\begin{align*}
&\mathbb{P}\left(\beta_{im}(t)-\b_{i\ell}(t)>  r \right) \cr
&\qquad \le 
e^{-\left((E_\LL-E_\M)\sum_{k=0}^t p_m(k)+ r\right)^2(t+1)^{-1}/2}\cr
&\qquad\le e^{-r^2(t+1)^{-1}/2}.
\end{align*}
\noindent(b) \ 
The result follows from the proof of part (a) when neighbor $j$ is malicious. In this case, since $r<0$, we must ensure that $\lambda^*>0$. When $(E_\LL-E_\M) \sum_{k=0}^t p_m(k)+r>0$ for some $t\ge0$, this condition is satisfied.
\end{proof}

\subsection{Detection Analysis}\label{sec:detection_analysis}
Here, we present our main results on the misclassification probabilities of \Cref{alg_trust_neig}. The following two results show that misclassification of legitimate and malicious neighbors (see \Cref{fig:alg_intuition}) decay at a near-geometric rate.
\begin{lemma} \label{lem_misp_legit}
 Let Assumption~\ref{leg_ag} and Assumption~\ref{mal_ag} hold. Let $j$ be a legitimate neighbor of a legitimate agent $i$, i.e., $j \in \Nil$. Then, for any $t\ge0$, the misclassification probability that agent $i$ excludes its legitimate neighbor $j\in\Nil$ from the trusted neighborhood $\hat\Ni(t)$
has the following upper bound:
\begin{align*}
    \mathbb{P}(&j \not \in \hat{\ccalN}_i(t)) \leq |\ccalN_i |\cdot 
    e^{-\xi_t^2(t+1)^{-1}/2}.
\end{align*}
\end{lemma}
\begin{proof}
      \Cref{alg_trust_neig} ensures that $j \notin \hat{\ccalN}_i(t)$ occurs if and only if the condition ${\beta_{i\bar{j}(t)}(t) - \beta_{ij}(t) \le \xi_t}$ is not met. This misclassification happens when there exists at least one agent $n \in \Ni$ such that ${\beta_{in}(t) - \beta_{ij}(t) > \xi_t}$ holds. Apparently, this can happen only for an agent $n\ne j$. The misclassification event can be characterized as follows:
     \begin{equation*}
              \{\beta_{i\bar{j}(t)}(t)-\beta_{ij} (t) > \xi_t\}
              \!=\!\bigcup_{n \in \Ni\setminus\{j\}} \{ \beta_{in}(t)-\beta_{ij}(t) > \xi_t \}.
     \end{equation*}
     Thus, we have
    \begin{align*}
        \mathbb{P}(j \not \in \hat{\ccalN}_i(t)) 
        = & \mathbb{P}(\beta_{i\bar{j}(t)}(t)-\beta_{ij} (t) > \xi_t) \\
        = & 
        \mathbb{P}(\bigcup_{n \in \Ni\setminus\{j\}} 
        \{ \beta_{in}(t)-\beta_{ij}(t) > \xi_t \}) \\
        \leq & 
        \sum_{n \in \Ni\setminus\{j\}} \mathbb{P}( \beta_{in}(t)-\beta_{ij}(t) > \xi_t) 
        \label{eq_j_legitimate_union_bound}
    \end{align*}
    By Lemma~\ref{lem-opt-bound}(a), where we let $r=\xi_t$,
    we have for any $n\in\Ni$,
\[\mathbb{P}\left(\beta_{in}(t)-\b_{i\ell}(t)> \xi_t\right)\le e^{-\xi_t^2(t+1)^{-1}/2}.\]
The stated result follows by summing these bounds over $n$.
\end{proof}

\begin{lemma} \label{lem_misp_mal}
  Suppose Assumption~\ref{leg_ag} and~\ref{mal_ag} hold. Let $m$ be an arbitrary malicious neighbor of a legitimate agent $i$, i.e., $m \in \Nim$ for agent $i \in \mathcal{L}$. Then, for all times $t\ge0$ such that $(E_\LL-E_\M) \sum_{k=0}^t p_m(k)-\xi_t>0$, the misclassification probability of agent $m$ by agent $i$ has the following upper bound:
   \begin{align*}
     \mathbb{P}(m\in \hat{\ccalN}_i(t))
      \leq e^{-\left((E_\LL-E_\M) \sum_{k=0}^t p_m(k)-\xi_t\right)^2(t+1)^{-1}/2}.
 \end{align*}
\end{lemma}
\begin{proof}
    Misclassification of a malicious agent $m \in \mathcal{M}$ occurs when it remains within the trusted neighborhood, represented by the event $\beta_{i\bar{j}(t)}(t) - \beta_{im}(t) \leq \xi_t$ as per Algorithm \ref{alg_trust}. For a malicious agent to be mislassified, its accumulated trust value $\beta_{im}(t)$ must be less than $\xi_t-\beta_{i\bar{j}(t)}(t)$ where $\bar{j}(t)$ is the most trusted agent. This condition is satisfied if and only if $\beta_{im}(t)$ is less than $\xi_t-\beta_{in(t)}(t)$ for all neighbors $n \in \Ni$. Such requirement leads to the formulation that can be expressed as the intersection of pairwise comparisons:   
    \begin{equation}
              \{\beta_{i\bar{j}(t)}(t)-\beta_{im} (t) \leq \xi_t\}=\bigcap_{n \in \Ni} \{ \beta_{in}(t)-\beta_{im}(t) \leq \xi_t \}.
    \end{equation}
    Therefore, the misclassification probabilities of these events are also equivalent, and we bound the probability of the intersection of events with the minimum probability over the given set of events,
    \begin{align*}
        \mathbb{P}(m\in \hat{\ccalN}_i(t)) &= \mathbb{P}(\beta_{i\bar{j}(t)}(t)-\beta_{im}(t) \leq \xi_t) \\
        &= \mathbb{P}(\bigcap_{n \in \Ni} \{ \beta_{in}(t)-\beta_{im}(t) \leq \xi_t \}) \\
        &\leq \mathbb{P}( \beta_{il}(t)-\beta_{im}(t) \leq \xi_t), \text{ where $l\in \Nil$}\\
        &=\mathbb{P}( \beta_{im}(t)-\beta_{il}(t) \geq -\xi_t),
    \end{align*}
Note that Assumption~\ref{leg_ag}(1) ensures the existence of a legitimate neighbor $l\in \Nil$ for any $i\in \LL$. Applying Lemma~\ref{lem-opt-bound}(b) with $r=-\xi_t$, we obtain
\begin{align*}
    \mathbb{P}( \beta_{im}(t)-\beta_{il}(t) &\geq -\xi_t) \\
    &\leq e^{-\left((E_\LL-E_\M) \sum_{k=0}^t p_m(k)-\xi_t\right)^2(t+1)^{-1}/2}.
\end{align*}
\end{proof}
Lemma \ref{lem_misp_legit} shows that the misclassification probability of legitimate neighbors converges to $0$ as $\xi_t \rightarrow \infty$. Faster convergence can be achieved by choosing a threshold $\xi_t$ to grows more quickly. Intuitively, a larger threshold increases the likelihood that a neighbor---whether legitimate and malicious---will be included in the trusted neighborhood, thereby reducing the chance of excluding legitimate ones. On the other hand, Lemma~\ref{lem_misp_mal} shows that the misclassification probability of malicious neighbors can also converge to zero, but under different conditions. The convergence rate depends positively on how frequently malicious agents attack, quantified by the term $(E_\LL-E_\M) \sum_{k=0}^t p_m(k)$, and inversely on $\xi_t$. Therefore, increasing $\xi_t$ can slow down the detection of malicious agents, since a larger threshold increases the likelihood of including any neighbor, as previously discussed. In the next part, we derive sufficient conditions on the threshold $\xi_t$ and the minimum conditional attack probabilities $p_m(k)$ to ensure that malicious agents are detected almost surely.

\subsection{Convergence To Consensus}\label{sec:convergence_of_consensus}


In the previous section, we derived bounds on the probability of misclassifying neighbors as a function of the threshold $\xi_t$ and the minimum conditional attack probabilities $p_m(k)$. We now use these bounds to show the almost surely convergence of our consensus algorithm. The key idea is to choose $\xi_t$ so that all misclassification events cease to occur after a finite time. We begin by formally defining this notion of finite-time correctness.
\begin{definition}[$T_f$]\label{def_tf}
    Define the event that any legitimate agent $i\in \LL$ misclassifies a legitimate neighbor $l\in \Nil$ at time $t$ as
    \begin{equation}
        \mathcal{A}_l(t) := \bigcup_{i \in \mathcal{L}} \bigcup_{l \in \Nil} \{l \notin \hat{\ccalN}_i(t) \}.
    \end{equation}
    Similarly, define the event that any legitimate agent $i\in \LL$ misclassifies a malicious neighbor $m\in \Nim$ at time $t$ as
    \begin{equation}
        \mathcal{A}_m(t) := \bigcup_{i \in \mathcal{L}} \bigcup_{m \in \Nim} \{m \in \hat{\ccalN}_i(t) \}.
    \end{equation} 
    If there exists a random but finite time $T_f$ such that for all $t\geq T_f$, both misclassification events no longer occur, i.e., $\mathcal{A}_l(t)=\emptyset$ and $\mathcal{A}_m(t)=\emptyset$, we say that all neighbors are correctly classified from time $T_f$ onward. Moreover, if such a time exists, we define
    \begin{equation}\label{eq:Tf}
        T_f := \inf\left\{ t \geq 0 \;\middle|\; \forall k \geq t,\; \mathcal{A}_l(k) = \emptyset,  \mathcal{A}_m(k) = \emptyset \right\},
    \end{equation}
     which represents the earliest time after which no further misclassifications occur.
\end{definition}
\begin{remark}
    If such a random time $T_f$ exists ($T_f \not = \infty$), then for all $t\geq T_f$, no legitimate agent excludes any legitimate neighbor or includes any malicious neighbor. Consequently, the weight matrices satisfy $W_{\ccalL} (t)= \widebar{W}_{\ccalL}$ for all $t\geq T_f$.
\end{remark}
Our first goal is to provide sufficient conditions for the existence of $T_f$. The following assumptions will provide such sufficiency conditions.
\begin{assumption}[Threshold and Attack Probabilities]
\label{asmp_th}
There exists a time $t' \ge 0$ and a constant $\epsilon>0$ such that, for all $t \ge t'$ and $m\in \M$:\\ 
\noindent (1) $\xi_t \geq \sqrt{(1+\epsilon)(t+1)\ln(t+1)},$\\ 
\noindent (2) $(E_\LL-E_\M) \sum_{k=0}^t p_m(k)\geq \xi_t+\sqrt{(1+\epsilon)(t+1)\ln(t+1)}.$
\end{assumption}
\begin{lemma}\label{lem_Tf}
    Let Assumption~\ref{leg_ag} and  Assumption~\ref{mal_ag} hold. \\
    \noindent (1) Suppose Assumption~\ref{asmp_th}(1) holds. Then the event $\mathcal{A}_l(t)$, in which some legitimate agent misclassifies a legitimate neighbor, occurs only finitely many times almost surely. \\
    \noindent (2) Suppose Assumption~\ref{asmp_th}(2) holds. Then the event $\mathcal{A}_m(t)$, in which some legitimate agent misclassifies a malicious neighbor, occurs only finitely many times almost surely. \\
    \noindent (3) If both Assumptions~\ref{asmp_th}(1) and \ref{asmp_th}(2) hold, then there exists a (random) finite time $T_f$ such that every legitimate agent $i\in \LL$ classifies all of its neighbors correctly almost surely. Moreover, we have $W_{\ccalL} (t)= \widebar{W}_{\ccalL}$ almost surely for all $t\geq T_f$.
\end{lemma}
\begin{proof}
We start with part (1) and focus on the event $\mathcal{A}_l(t)$. By the first Borel–Cantelli lemma, it suffices to show $\sum_{t=0}^{\infty} \mathbb{P}(\mathcal{A}_l(t)) < \infty$. The event $\mathcal{A}_l(t)$ is defined as finite unions over agent pairs. Therefore, it is enough to verify $\sum_{t=0}^{\infty} \mathbb{P}(\{l \notin \hat{\ccalN}_i(t) \}) < \infty$ for a single pair $(i,l)$ with $i\in \LL$, $l\in \Nil$ provided that the corresponding bounds apply uniformly to all such pairs. By \Cref{lem_misp_legit}, we have $\mathbb{P}(\{m \notin \hat{\ccalN}_i(t) \}) \leq |\ccalN_i |\cdot 
    e^{-\xi_t^2(t+1)^{-1}/2}$, where the constants $|\ccalN_i |$ and $e^{-1/2}$ do not affect convergence. Thus, we study the series $\sum_{t=0}^{\infty} e^{-\xi_t^2(t+1)^{-1}}$. For any $\epsilon>0$, if $\xi_t \geq \sqrt{(1+\epsilon)(t+1)\ln(t+1)}$ for sufficiently large $t$, then,
    \begin{align*}
        e^{-\xi_t^2(t+1)^{-1}} \leq e^{-(1+\epsilon)\ln(t+1)} = 1/(t+1)^{(1+\epsilon)}.
    \end{align*}
    Since $\sum_{t=0}^{\infty} 1/(t+1)^{(1+\epsilon)}<\infty$ for any $\epsilon>0$, by the comparison test we get $\sum_{t=0}^{\infty} e^{-\xi_t^2(t+1)^{-1}}<\infty$ \cite[Corollary~7.3.2, pg.148]{taoAnalysis2022}. Hence, $\sum_{t=0}^{\infty} \mathbb{P}(\mathcal{A}_l(t)) < \infty$, and the Borel-Cantelli lemma ensures that $\mathcal{A}_l(t)$ occurs only finitely often almost surely. 
    
For part (2), Lemma~\ref{lem_misp_mal} shows that the probability of misclassifying a malicious neighbor is bounded by $\mathbb{P}(m\in \hat{\ccalN}_i(t))
      \leq e^{-\left((E_\LL-E_\M) \sum_{k=0}^t p_m(k)-\xi_t\right)^2(t+1)^{-1}/2},$ provided that $(E_\LL-E_\M) \sum_{k=0}^t p_m(k)>\xi_t$. Under the assumed lower bound on $(E_\LL-E_\M) \sum_{k=0}^t p_m(k)-\xi_t$ in Assumption~\ref{asmp_th}(2), this  condition is satisfied. Then, by a comparison argument similar to part (1), the series $\sum_{t=0}^{\infty} \mathbb{P}(\{l \notin \hat{\ccalN}_i(t) \}) < \infty$ is summable for all pairs $(i,m)$ with $i\in \LL$, $m\in \Nim$. 
      
Finally, for part (3), if both Assumptions~\ref{asmp_th}(1) and~\ref{asmp_th}(2) hold, the result follows from parts (1) and (2), together with the definition of $T_f$ in \Cref{def_tf}.
\end{proof}
\begin{remark}
Since the conditional attack probabilities satisfy $p_m(k) \in [0,1]$, it follows that $\sum_{k=0}^t p_m(k) \le (t+1)$.
Therefore, $(E_{\mathcal{L}} - E_{\mathcal{M}})\sum_{k=0}^t p_m(k)$ can grow linearly in $t$ if each $p_m(k)$ has a non-zero lower bound.
\end{remark}
We note that the conditions provided by \Cref{asmp_th} are only sufficiency conditions. There could be arbitrarily many other thresholds and attack probabilities under which time $T_f$ exists almost surely as there exists neither a greatest convergent sum of sequences nor a smallest divergent sum of sequences \cite{ash1997neither}. Still, \Cref{lem_Tf} (and \cref{asmp_th}) encompasses a variety of threshold schedules $\xi_t$ and attack probability sequences $\{p_m(t)\}$, generalizing prior works \cite{yemini2021characterizing,aydin2024multi}. In \cite{aydin2024multi}, each agent chooses $\xi_t = \xi (t+1)^{\gamma}$ for some constants $\xi>0$ and $\gamma\in (0.5,1)$. Moreover, it is assumed there exists a uniform lower bound $\bar{p}>0$ such that $p_m(t)\geq \bar{p}$ for all $t$. Under these conditions, $(E_{\mathcal{L}} - E_{\mathcal{M}})\sum_{k=0}^t p_m(k)$ grows \emph{linearly} in $t$, whereas $\xi_t + \sqrt{(t+1)\ln(t+1)}$ grows only \emph{sublinearly}. Consequently, \Cref{asmp_th} is satisfied, guaranteeing a finite time $T_f$ after which no agent misclassifies any neighbor. The extreme case in \cite{yemini2021characterizing} where malicious agents always attack (i.e., $p_m(k)=1$), can also be covered by an appropriate choice of $\xi_t$ (for example $\xi_t = \sqrt{(1+\epsilon)(t+1)\ln(t+1)}$). Furthermore, if $(E_{\mathcal{L}} - E_{\mathcal{M}})$ is known, an even stronger choice such as $\xi_t=(E_{\mathcal{L}} - E_{\mathcal{M}})(t+1)/2$ can yield geometric decay in both legitimate and malicious misclassification probabilities. Beyond these examples, more gradual threshold schedules are likewise possible. For instance, if $p_m(k)\geq \bar{p}>0$, one may set $\xi_t \geq \sqrt{(1+\epsilon)(t+1)\ln(t+1)}$, which grows more slowly than $\xi_t = \xi (t+1)^{\gamma}$ for $\gamma \in (0.5,1)$ yet still satisfies the condition. Taken together, these cases illustrate the flexibility of \Cref{asmp_th} in encompassing diverse scenarios with varying threshold growth rates and attack strategies.

\begin{remark}
Note that our results do not require coordinated threshold selection among legitimate agents; each legitimate agent $i\in \LL$ may choose its own threshold $\xi_t$ independently as long as the conditions in \Cref{asmp_th} are satisfied. The proofs extend naturally to this uncoordinated setting.
\end{remark}

Next, we shift our focus to the analysis of our consensus algorithm by leveraging the results on $T_f$.

\begin{lemma} \label{lem_als_w}
    Suppose Assumptions \ref{leg_ag}, \ref{mal_ag}, and~\ref{asmp_th} hold. Then, it holds almost surely
    \begin{equation}
        \prod_{t=T_0-1}^\infty \wl (t)= \mathbf{1}\nu^T \bigg (  \prod_{t=T_0-1}^{\max \{T_f,T_0\}-1} \wl (t) \bigg ),
    \end{equation}
    where the matrix product $\prod_{t=T_0-1}^\infty \wl (t) >{\bf 0}$ for any $T_0\ge0$ almost surely, and $\nu >{\bf 0}$ is a stochastic vector. 
\end{lemma}
\begin{proof}
    The result follows from Proposition 2 of \cite{yemini2021characterizing} by the existence of the finite time $T_f$ (\Cref{lem_Tf}), i.e., $W_{\ccalL} (t)= \widebar{W}_{\ccalL}$ for all $t \ge T_f$.
\end{proof}
We provide the results on the limit behavior of consensus process using the almost sure convergence of weight matrices. 
\begin{lemma} \label{lem_con_legit}
    Suppose Assumptions \ref{leg_ag}, \ref{mal_ag}, and \ref{asmp_th} hold. Given the initial values of legitimate agents $\xl(0)$, the process $\Tilde{x}_{\ccalL}(T_0,t)$ converges almost surely, i.e., almost surely
    \begin{equation*}
            \lim_{t \rightarrow \infty}  \Tilde{x}_{\ccalL}(T_0,t)= \bigg ( \prod_{k=T_0-1}^{\infty} \wl(k) \bigg)\xl(0) = y\mathbf{1},
    \end{equation*}
where $y \in \mathbb{R}$ is a random variable depending on $T_f$ and $T_0$. 
\end{lemma}

\begin{proof}
    The result is an immediate consequence of Proposition~2 in \cite{yemini2021characterizing}. 
\end{proof}

We next state the limit of $\phi_{\ccalM}(T_0,t)$ (see~\eqref{eq_sep_Dy}) in the consensus process.

\begin{lemma} \label{lem_con_mal}
    Suppose Assumptions \ref{leg_ag}, \ref{mal_ag}, and \ref{asmp_th} hold. Then, the effect of malicious agents $\phi_{\ccalM}(T_0,t)$ converges almost surely, i.e., we have almost surely
    \begin{align*}
            \lim_{t \rightarrow \infty}\phi_{\ccalM}(T_0,t)&=  \sum_{k=T_0-1}^{\infty} \bigg ( \prod_{\ell=k+1}^{\infty} \wl(\ell) \bigg) \wm(k) \xm(k)\\
            &= h\mathbf{1},
    \end{align*}
where $h \in \mathbb{R}$ is a random variable depending on $T_f$ and $T_0$.
\end{lemma}

\begin{proof}
The result and its proof are identical to those of Proposition 3 in \cite{yemini2021characterizing}, derived from the almost sure convergence of the weight matrices, $W_{\ccalL} (t)= \widebar{W}_{\ccalL}$ for all $t \ge T_f$.
\end{proof}

Next, our final result in this part states that legitimate agents can still reach an agreement but over a random value asymptotically.
\begin{corollary} \label{cor_con}
    Suppose Assumptions \ref{leg_ag}, \ref{mal_ag}, and \ref{asmp_th} hold. Then, the consensus protocol~\eqref{eq_con_dy} among the legitimate agents  converges almost surely, i.e.,
    \begin{equation}
         \lim_{t \rightarrow \infty}  \xl(T_0,t)=  z\mathbf{1}\quad\hbox{almost surely},
    \end{equation}
where $z \in \mathbb{R}$ is a random variable given by $z=y+h$, with $y$ and $h$ from Lemma~\ref{lem_con_legit} and Lemma~\ref{lem_con_mal}, respectively.
\end{corollary}
\begin{proof}
    The result is a direct consequence of Lemmas \ref{lem_con_legit}--\ref{lem_con_mal} (Propositions 2-3 in \cite{yemini2021characterizing}).
\end{proof}

Corollary~\ref{cor_con} indicates that the legitimate agents reach the same random scalar value implying $\lim_{t\rightarrow\infty} |x_i(t)- x_j(t)|=0$ almost surely for any $(i,j) \in \mathcal{L} \times \mathcal{L}$. Note that Corollary~\ref{cor_con} does not guarantee that the consensus value $z$ lies within the convex hull of the legitimate agents' initial values $\xl
(0)$ due to the influence of malicious agents.

\begin{corollary}[Convergence in mean] \label{cor_con_exp}
    The legitimate agents' values converge to the same value in expectation, i.e.,
\begin{equation}
    \lim_{t \rightarrow \infty} \mathbb{E} (\|
\xl(T_0,t)-z\mathbf{1}\|)=0  
\end{equation}
    
\end{corollary}

\begin{proof}
  By Corollary \ref{cor_con},  almost surely  we have $\lim_{t \rightarrow \infty}  \xl(T_0,t)=  z\mathbf{1}$, where $z \in \mathbb{R}$ is a random variable. Since $\|\xl(T_0,t)\| \le \eta $ for any time $t$,  the result follows by Lebesgue Dominated Convergence Theorem.

\end{proof}

Next, we characterize the probability of reaching $T_f$ and the deviation in the consensus process from the nominal case without malicious agents. These characterizations depend on the rate at which legitimate agents correctly classify their neighbors, as captured by the upper bounds in \Cref{lem_misp_legit} and \Cref{lem_misp_mal}. Notably, more frequent attacks by malicious agents lead to faster detection, as reflected in \Cref{lem_misp_mal}. To capture the worst-case scenario from a detection standpoint, we focus our analysis on the case where the lower bounds in \Cref{asmp_th} are attained.

\subsection{Characterizing $T_f$}\label{sec:char_tf}
We make the following assumption, which will be used instead of \Cref{asmp_th} throughout the rest of the paper.

\begin{assumption}
\label{asmp_th_lb}
Let $\epsilon_1>0$ be a constant. Moreover, let $\epsilon_2>0$ be a constant such that $\sqrt{1+\epsilon_2}\geq\frac{2\sqrt{1+\epsilon_1}}{E_\LL-E_\M}$. Assume the following hold for all $t\geq0$ and $m\in \M$:\\ 
\noindent (1) $\xi_t = \sqrt{(1+\epsilon_1)(t+1)\ln(t+1)},$\\ 
\noindent (2) $\sum_{k=0}^t p_m(k)\geq \sqrt{(1+\epsilon_2)(t+1)\ln(t+1)}.$
\end{assumption}
This assumption ensures that the lower bounds in \Cref{asmp_th} are attained. For clarity of exposition, we impose these conditions for all $t\geq 0$, although the analysis naturally extends to scenarios where they hold for $t\geq t'$ for some $t'>0$.
\begin{remark}\label{remark_pm_lb}
    \Cref{asmp_th_lb} captures scenarios in which the lower bound on the cumulative attack probability of malicious agents converges to zero, provided the decay is not too sharp. For instance, the analysis extends to cases such as $p_m(t)=1/(t+1)^c$ with $c<1/2$, where $p_m(t)$ converges to $0$ but $T_f$ still exists.
\end{remark}
Prior works such as \cite{yemini2021characterizing,ballotta2024role} focus on the extreme case where malicious agents attack at every step, i.e., $p_m(t)=1$, while our previous conference paper \cite{aydin2024multi} considers attack probabilities lower bounded by a positive constant, i.e., $p_m(t)\geq\bar{p}>0$. In contrast, this work does not impose such lower bounds, as discussed in \Cref{remark_pm_lb}. Now, using this, we analyze the probability of reaching $T_f$ at some time $t$. First, we need the following lemma regarding the infinite summation of misclassification probabilities.
\begin{lemma} \label{lem_hurwitz}
    Suppose Assumptions \ref{leg_ag}, \ref{mal_ag}, and \ref{asmp_th_lb} hold. Let $\zeta(c,t)$ denote the Hurwitz zeta function defined by $\zeta(c,t):=\sum_{k=0}^{\infty} \frac{1}{(k+t)^{c}}$ where $c>0$ and $t>0$. Let $i\in \LL$ be a legitimate agent with a legitimate neighbor $l\in \Nil$ and a malicious neighbor $m\in \Nim$. Then, for all $t\geq0$ we have
    \begin{align*}
        \sum_{k=t}^\infty \mathbb{P} ( l \not \in \hat{\ccalN}_i (k)) &\leq \zeta(1+\epsilon_1,t), \text{ and} \\ \sum_{k=t}^\infty \mathbb{P}( m  \in \hat{\ccalN}_i (k)) &\leq \zeta(1+\epsilon_2,t),
    \end{align*}
    where $\epsilon_1$ and $\epsilon_2$ are the constants defined in \Cref{asmp_th_lb}.
\end{lemma}
\begin{proof}
    The proof directly follows from Lemma~\ref{lem_misp_legit} and Lemma~\ref{lem_misp_mal} with the choice of $\xi_t$ and $p_m(t)$ stated in Assumption~\ref{asmp_th_lb}.
\end{proof}
\begin{proposition} \label{prop_Tf_prob}
    Suppose Assumptions \ref{leg_ag}, \ref{mal_ag}, and \ref{asmp_th_lb} hold. Let $\zeta(c,t)$ denote the Hurwitz zeta function defined by $\zeta(c,t):=\sum_{k=0}^{\infty} \frac{1}{(k+t)^{c}}$ where $c>0$ and $t>0$. 
    The probability of the event that all agents are correctly classified after time step $t \in \mathbb{N}$ is bounded below as follows, i.e, 
    \begin{align}
    &\mathbb{P}(T_f =t) \le|\ccalL|^2 |\ccalN|\cdot 
    e^{-\xi_t^2(t+1)^{-1}/2} \nonumber\\
    &+ |\ccalL| |\ccalM|\cdot \: e^{-\left((E_\LL-E_\M) \cdot \underset{m \in \ccalM}{\min} \sum_{k=0}^t \ p_m(k)-\xi_t\right)^2(t+1)^{-1}/2},
\end{align}
and
\begin{align}
    \mathbb{P}(T_f >t-1) \le |\ccalL|^2 |\ccalN|\cdot \zeta(1+\epsilon_1,t) +|\ccalL| \cdot |\ccalM| \zeta(1+\epsilon_2,t).
    \end{align}
\end{proposition}
\begin{proof}
We first define the event that there are no misclassified neighbors of legitimate agents at time $t$. Therefore, this event can be expressed as the intersection of events of correct classification for each legitimate agent $i \in \ccalL$ and for all of its neighbors $\Ni=\Nil\cup \Nim$, 
\begin{align}
    \mathcal{D}(t):= \left \{ \bigcap_{\substack{i\in \mathcal{L} \\  l \in \Nil}} \{ l   \in \hat{\ccalN}_i (t) \} \bigcap_{\substack{i \in \mathcal{L} \\  m \not \in \Nim}} \{ m  \not \in \hat{\ccalN}_i (t) \right \}.
\end{align}
Then we have, by the definition of the earliest time step $T_f$ (Eq. \eqref{eq:Tf} where misclassification of agents no longer happens, due to the fact that 
    \begin{align*}
        \mathbb{P}(T_f=t)&= \mathbb{P} \left (\left \{ \bigcap_{k \ge t} \mathcal{D}(k) \right \} \cap \mathcal{D}^C(t-1) \right).
    \end{align*}
    
Next, we derive the upper bound for the union of misclassification agents,      

\begin{align}
        &\mathbb{P}(T_f=t) \le \mathbb{P}(\mathcal{D}^C (t-1)) \nonumber\\
        &=\mathbb{P}\left(\bigcup_{\substack{i\in \mathcal{L} \\  l \in \Nil}} \{ l   \not \in \hat{\ccalN}_i (t) \} \bigcup_{\substack{i \in \mathcal{L} \\  m \not \in \Nim}} \{ m  \in \hat{\ccalN}_i (t) \right) \nonumber\\
        &\le \sum_{i \in \ccalL}  \left(\sum_{l \in \Nil} \mathbb{P}
        (l \not \in \hat{\ccalN}_i (t)) + \sum_{m \in \ccalM}\mathbb{P}(m  \in \hat{\ccalN}_i (t)\right) \label{eq_upper_prob_bound}.   
    \end{align}
Hence, the result follows from by Lemmas \ref{lem_misp_legit} and \ref{lem_misp_mal},
\begin{align}
    \mathbb{P}(T_f=t) &\le|\ccalL|^2 |\ccalN|\cdot 
    e^{-\xi_t^2(t+1)^{-1}/2} \nonumber\\
    &+ |\ccalL| |\ccalN| \cdot e^{-\left((E_\LL-E_\M) \underset{m \in \ccalM}{\min} \sum_{k=0}^t p_m(k)-\xi_t\right)^2(t+1)^{-1}/2}.
\end{align}

Using the union bound in Eq. \eqref{eq_upper_prob_bound} and Lemma \ref{lem_hurwitz}, we also conclude that the upper bound for the probability $\mathbb{P}(T_f > t-1)$ as follows,
\begin{align}
    \mathbb{P}(T_f > t-1)&\leq \sum_{k=t}^\infty \mathbb{P}(T_f =k) \nonumber\\
    &\le |\ccalL|^2 |\ccalN|\cdot \zeta(1+\epsilon_1,t) +|\ccalL| \cdot |\ccalM| \zeta(1+\epsilon_2,t). \label{eq_sum_upper_Tf}
\end{align}
\end{proof}

\subsection{Deviation from Nominal Consensus}\label{sec:deviation}

Throughout this section, we aim to characterize the deviation from the asymptotic nominal consensus value. The nominal consensus dynamics represent the ideal scenario in which $\wl(t) = \barwl$ and $\wm(t) = \mathbf{0}$ for all $t \geq T_0 - 1$. Consequently, the asymptotic nominal consensus value, which denotes the ideal state the agents would achieve, is given by $\mathbf{1} \nu^T x_{\ccalL}(0)$, since $\lim_{t \to \infty} \barwl^t = \mathbf{1} \nu^T$. As we do not have any assumptions on the dynamics of malicious agents in Eq.~\eqref{eq_con_dy}, we will analyze the process with the worst-case approach, based on the idea that legitimate agents stop assigning positive weights to malicious agents after some (random) finite time. We first bound the probability that legitimate agents do not follow the nominal consensus dynamics with the nominal weights $\barwl$ after the observation window $T_0$. 

\begin{lemma} \label{lem_WlnotbarW}
  Suppose Assumptions \ref{leg_ag},\ref{mal_ag} and \ref{asmp_th_lb} hold. The probability of the event that the actual consensus dynamics among legitimate agents deviate from the nominal dynamics at some time step $k \ge T_0 - 1$ is bounded as follows,
\begin{align*}
    &\mathbb{P} ( \exists k \ge T_0-1: \wl(k) \neq \widebar{W}_{\ccalL} ) \nonumber\\
     & \le |\ccalL|^2 |\ccalN|\cdot \zeta(1+\epsilon_1,T_0-1) +|\ccalL| \cdot |\ccalM| \zeta(1+\epsilon_2,T_0-1).
\end{align*}
\end{lemma}

\begin{proof}
    The event in which the actual weight matrix differs from the nominal weight matrix,  $\wl(k)\not = \barwl$ at some time $k\ge T_0-1 $ is equivalent to the event that at some time $t \ge T_0-1$, there exists an agent misclassified by a legitimate agent. Therefore, we have $\mathbb{P} ( \exists k \ge T_0-1: \wl(k) \neq \widebar{W}_{\ccalL} )=\mathbb{P} (\bigcup_{k \ge T_0-1} \mathcal{D}^C (k)) $ and the bound follows from the union bounds over the individual misclassification events over time as in Eqs. \eqref{eq_upper_prob_bound} and \eqref{eq_sum_upper_Tf} in Lemma \ref{lem_hurwitz}.
    
\end{proof}

In Lemma \ref{lem_WlnotbarW}, we bounded the deviation from the nominal consensus dynamics by expressing it as the union of misclassification events by legitimate agents. Now, we derive the deviation resulting from the difference between actual and nominal weights matrices 
($\wl(t)$ and $\barwl$ in order) of legitimate agents.
\begin{lemma} \label{lem_dev_p1}
    Suppose Assumptions \ref{leg_ag},\ref{mal_ag} and \ref{asmp_th_lb} hold. Let $\varphi_i(T_0,t)$ be a deviation experienced by a legitimate agent $i \in \mathcal{L}$, stemming from the difference between actual and nominal weights of legitimate agents over time,  defined formally as follows, for all $t\ge 0$,
    \begin{equation}\label{eq_dev}
        \varphi_i(T_0,t):= \Bigg| \Bigg [ \Tilde{x}_{\ccalL}(T_0,t)-\bigg ( \prod_{k=T_0-1}^{t-1} \barwl \bigg)\xl(0) \Bigg]_i \Bigg|.
    \end{equation}
Then, for an error level $\delta>0$, we have
\begin{align*}
    \mathbb{P} \Big ( \max_{i \in \mathcal{L}} \: \limsup_{t \rightarrow \infty} \varphi_i(T_0,t) > \frac{2\eta}
    {\delta} g_{\mathcal{L}} (T_0) \Big ) < \delta , 
\end{align*}
where $\eta \ge \sup_{i \in \mathcal{N}, t \in \mathbb{N} } |x_i(t)|$, we define 
\begin{align}
 g_{\mathcal{L}} (T_0)&:=  |\ccalL|^2 |\ccalN|\cdot \zeta(1+\epsilon_1,T_0-1) \\
 &+|\ccalL| |\ccalM|\cdot  \zeta(1+\epsilon_2,T_0-1).
 \label{eq_gl}
\end{align}
\end{lemma}
\begin{proof}
    The proof is a refinement of Proposition 4 in \cite{yemini2021characterizing} by implementing adjustments to the lower bound on the probability of the given event, and the starting time $T_0$.
\end{proof}
The result of Lemma \ref{lem_dev_p1} is the consequence of the probability of the event we defined in Lemma \ref{lem_WlnotbarW}, and monotone properties of the upper bounds on the deviation such that we analyze the probability of the given deviation using Markov's inequality. Next, we analyze the other part of the deviation resulting from the direct involvement of malicious agents in the consensus dynamics. We define the following term for each $i\in\ccalL$,
\begin{equation}\label{eq-mal-infl}
\phi_i(T_0,t)=\eta \sum_{k=T_0-1}^{t-1} \sum_{ j \in \ccalM} \Bigg [ \Bigg ( \prod_{\ell=k+1}^{t-1}   W_{\mathcal{L}} (\ell) \Bigg ) W_{\mathcal{M}} (k) \Bigg]_{ij}.
\end{equation}
The term $\phi_i(T_0,t)$ is an upper bound on the elements of the vector of malicious influence $\phi_{\mathcal{M}}(T_0,t)$ as defined in~\eqref{eq_sep_Dy},
\begin{equation*}
|[\phi_{\mathcal{M}}(T_0,t)]_i| \le \max_{i \in \mathcal{L}} \phi_i(T_0,t).
\end{equation*}

\begin{lemma} \label{lem_dev_p2}
Suppose Assumptions \ref{leg_ag},\ref{mal_ag} and \ref{asmp_th_lb} hold. For an error level $\delta > 0$, we have the following,
\begin{equation*}
\mathbb{P} \Bigg ( \max_{i \in \mathcal{L}} \limsup_{t \to\infty } \phi_i(T_0,t) > \frac{\eta }{ \kappa \delta} g_{\mathcal{M}}(T_0) \Bigg) < \delta
\end{equation*} 
where $\eta \ge \sup_{i \in \mathcal{N}, t \in \mathbb{N} } |x_i(t)|$
\begin{equation} \label{eq_gm}
    g_{\mathcal{M}} (T_0)=  |\ccalL| \cdot |\ccalM| \cdot \zeta(1+\epsilon_2,T_0-1).
\end{equation} 
\end{lemma}
\begin{proof}
We rewrite the event as the union over the set of agents,
\begin{align*}
&\mathbb{P} \Bigg ( \max_{i \in \mathcal{L}} \limsup_{t \to \infty } \phi_i(T_0,t) >\frac{\eta}{ \kappa \delta} g_{\mathcal{M}}(T_0) \Bigg )\\
&= \mathbb{P} \Bigg ( \bigcup_{i \in \mathcal{L}} \limsup_{t \to\infty } \phi_i(T_0,t) > \frac{\eta}{ \kappa \delta} g_{\mathcal{M}}(T_0) \Bigg ),
\end{align*}
The union bound and Markov's inequality provide the upper bound, as below,
\begin{align*}
&\mathbb{P} ( \max_{i \in \mathcal{L}} \limsup_{t \to \infty } \phi_i(T_0,t) > \frac{\eta}{ \kappa \delta} g_{\mathcal{M}}(T_0))\\
&\le\sum_{ i\in \mathcal{L}} \mathbb{P} ( \limsup_{t \to\infty } \phi_i(T_0,t) > \frac{\eta }{ \kappa \delta} g_{\mathcal{M}}(T_0)) \\
& \le\frac{ \delta\kappa |\mathcal{L}|\,\mathbb{E}(\limsup_{t \to \infty } \phi_i(T_0,t)) }{\eta g_{\mathcal{M}}(T_0)}.
\end{align*}
Next, we derive an upper bound for the expectation $\mathbb{E}(\limsup_{t \to\infty } \phi_i(T_0,t))$, starting with the upper bound for the random variable $\phi_i(T_0,t)$ in~\eqref{eq-mal-infl}), as follows,
\begin{align*}
\phi_i(T_0,t)
&\le \eta \sum_{k=T_0-1}^{t-1} \sum_{ j \in \ccalM} \frac{1}{\kappa} \Bigg(\sum_{n\in \ccalL} \tilde{w}_{in} \Bigg)
\end{align*}
where $\tilde{w}_{in}=[\prod_{l=k+1}^{t-1}   W_{\mathcal{L}} (\ell)]_{in}$ for $(i,n) \in \ccalL \times \ccalL$, and we used the fact that $[W_{\mathcal{M}} (k) ]_{ij} \le 1/\kappa$ for any $(i,j) \in \ccalL \times \ccalM$. The product of row-(sub)stochactic matrices is still row-(sub)stochactic, giving the property $\sum_{n\in \ccalL} \tilde{w}_{in} \le 1$. 
we rewrite the upper bound on $\bar{\phi}_i (T_0,t)$,  with the indicator variable $\mathbbm{1}_{\{j \in \hat{\ccalN}_i (t)\}}$, equal to $1$ when a malicious agent $j$ is included in the trusted neighborhood and otherwise $0$,
\begin{align*}
\phi_i(T_0,t)&
&\le  \frac{\eta}{\kappa} \sum_{k=T_0-1}^{t-1} \sum_{ j \in \ccalM} \mathbbm{1}_{\{j \in \hat{\ccalN}_i (t)\}} 
:=\bar{\phi}_i (T_0,t).
\end{align*}
The upper bound still holds for the expectation of limit superior in both sequences, i.e.,
\begin{equation*}
\mathbb{E} (\limsup_{t \to\infty } \phi_i(T_0,t)) \le \mathbb{E} (\limsup_{t \to\infty } \bar{\phi}_i(T_0,t)).
\end{equation*}
Since the random variables $\{\bar{\phi}_i (T_0,t)\}_{t\ge T_0}$ form a nonnegative and nondecreasing sequence as $t$ increases, we utilize the Monotone Convergence Theorem, and therefore have,
\begin{align*}
\mathbb{E} (\limsup_{t \to\infty } \bar{\phi}_i(T_0,t)) &= \mathbb{E} (\lim_{t \to \infty } \bar{\phi}_i(T_0,t)) \\
&= \lim_{t \to\infty }\mathbb{E} ( \bar{\phi}_i(T_0,t)).
\end{align*}

The properties, linearity of expectation, and expectation of indicators (equal to the probabilities of events defining the indicator variables) provide the following equivalence,
\begin{align*}
\lim_{t \to\infty } \mathbb{E}(  \bar{\phi}_i(T_0,t))
&=\frac{\eta}{\kappa}
\lim_{t \to \infty } \mathbb{E}\left(
\sum_{k=T_0-1}^{t-1} \sum_{ j \in \ccalM} \mathbbm{1}_{\{j \in \hat{\ccalN}_i (t)\}}\right)\\
&=\frac{\eta}{\kappa}
\lim_{t \to \infty } \sum_{k=T_0-1}^{t-1} \sum_{ j \in \ccalM} \mathbb{P}(j \in \hat{\ccalN}_i (t)).
\end{align*} 
Misclassification probabilities of malicious agents can be bounded by Lemma \ref{lem_misp_mal} and Lemma \ref{lem_hurwitz}, 
\begin{align*}
\mathbb{E}  &(\limsup_{t \to\infty } \phi_i(T_0,t)) \le 
\frac{\eta}{\kappa} 
\lim_{t \to \infty } \sum_{k=T_0-1}^{t-1} \sum_{ m \in \ccalM} \mathbb{P}(m \in \hat{\ccalN}_i (t))  \\
&\le \frac{\eta |\mathcal{M}|}{\kappa}  \cdot  \zeta(1+\epsilon_2,T_0-1).
\end{align*} 
Thus, for any error level $\delta >0$, the following bound is concluded by Markov's inequality,
\begin{align*}
&\mathbb{P} ( \max_{i \in \mathcal{L}} \limsup_{t \to \infty } \phi_i(T_0,t) >\frac{\eta }{ \kappa \delta} g_{\mathcal{M}}(T_0)) \\
& \le\frac{\delta\kappa|\mathcal{L}|\,\mathbb{E}(\limsup_{t \to \infty } \phi_i(T_0,t)) }{\eta g_{\mathcal{M}}(T_0)}\le \delta.
\end{align*}
\end{proof}

Lemma \ref{lem_dev_p2} concludes the upper bound on the probability of maximal deviation caused directly by malicious agents. The result relies on the convergent (infinite) sum of misclassification probabilities of malicious agents. Now, we present the final characterization of the deviation from the nominal consensus process.

\begin{theorem} \label{thm_dev}
Suppose Assumptions \ref{leg_ag},\ref{mal_ag} and \ref{asmp_th_lb} hold. For an error level $\delta >0$,
we have
\begin{align*}
    \mathbb{P}( \max \: \limsup_{t \xrightarrow[]{} \infty} |[ \xl(T_0,t)- \mathbf{1} \nu^T \xl(0) ]_i | &< \Delta_{\max} (T_0,\delta) ) \cr
    \ge 1-\delta,
\end{align*}
where $\Delta_{\max} (T_0,\delta)= 2(\frac{2\eta} {\delta} g_{\mathcal{L}}  (T_0)+ \frac{\eta} {\kappa \delta} g_{\mathcal{M}}(T_0))$.
\end{theorem}
\begin{proof}
    We conclude the theorem along the lines of Theorem 2 in \cite{yemini2021characterizing} incorporating the aforementioned modifications in the derived bounds, which are based on Lemmas \ref{lem_dev_p1}-\ref{lem_dev_p2}.
\end{proof}
In this section, we formally identified the bounds on the deviation from the nominal consensus process. Theorem \ref{thm_dev} follows from combining each part of the deviation derived in (Lemmas \ref{lem_dev_p1}-\ref{lem_dev_p2}). This result indicates that the starting time $T_0$ depends on the algorithmic parameters, and as agents start later (with increasing $T_0$), they have tighter and smaller bounds as a function of starting time $T_0$, and average trust difference $E_{\mathcal{L}}-E_{\mathcal{M}}$ between legitimate and malicious transmissions with the lower bound on the attack rate $\bar{p}$ in addition to the sequence of threshold parameters $\xi_{t}$ on the probability of deviations under the specified conditions.

\subsection{Convergence Rate}\label{sec:conv_rate}
For the analysis in this part, we firstly define a norm with respect to the stochastic vector $\nu \in ~\mathbb{R}^{|\ccalL|}$, $
    ||z||_{\nu} :=~ \sqrt{\sum_{i=1}^{\mathcal{|L|}} \nu_i z_i^2}.
$
\begin{theorem}[Convergence Rate of Consensus] \label{thm_rate}
   Suppose Assumptions \ref{leg_ag},\ref{mal_ag} and \ref{asmp_th_lb} hold. Then, we have, for  any $\tau \in \{ T_0-1,\cdots,t\}$
\begin{equation}
    || \xl (T_0,t)-  \mathbf{1}z||_\nu \le 2\eta(\tau-T_0+2)\rho_2^{t-\tau}.
\end{equation}
with a probability greater than,
\begin{align}
    &1-(|\ccalL|^2 |\ccalN|\cdot \zeta(1+\epsilon_1,T_0-1) +|\ccalL| \cdot |\ccalM| \zeta(1+\epsilon_2,T_0-1)).
\end{align}
where $\eta \ge \sup_{i \in \mathcal{N}, t \in \mathbb{N} } |x_i(t)|$, and $\nu \in \reals^{|\ccalL|}$ is a stochastic Perron vector of the matrix $\barwl$, \textit{i.e.,} $\nu ^T \barwl = \nu^T$.
\end{theorem}
\begin{proof}
    The result is a restatement of Theorem 3 in \cite{yemini2021characterizing} reflecting the changes in the lower bound on the probability.
\end{proof}
\begin{corollary}
    Suppose Assumptions \ref{leg_ag},\ref{mal_ag} and \ref{asmp_th_lb} hold. For any $\tau \in \{T_0-1,\cdots, t\}$ we have,
      \begin{align}
        &\mathbb{E} (  || \xl (T_0,t)-  \mathbf{1}z||_\nu ) \le \min_{\tau \in \{T_0-1,\cdots, t\}} 2\eta(\tau-T_0+2)\rho_2^{t-\tau} \nonumber\\
        &+ 2\eta (|\ccalL|^2 |\ccalN|\cdot \zeta(1+\epsilon_1,T_0-1) +|\ccalL| \cdot |\ccalM| \zeta(1+\epsilon_2,T_0-1)).
    \end{align}
\end{corollary}
\begin{proof}
 The result follows from the law of total expectation and the expectations conditioned on the event that the weight matrices $\wl(k)$  become equal to the nominal within a finite time horizon for $k \in \{T_0-1,\cdots, t\},$ as shown in Corollary 3 of \cite{yemini2021characterizing}.
\end{proof}

In this section, we formally analyze the finite-time performance of the consensus process. First, we establish the probability of the convergence rate. Then, we characterize the expected deviation from the consensus point within a finite time.
\section{Numerical Studies}
\begin{figure*}
	\centering
	\begin{tabular}{ccc}
	\includegraphics[width=.3\linewidth]{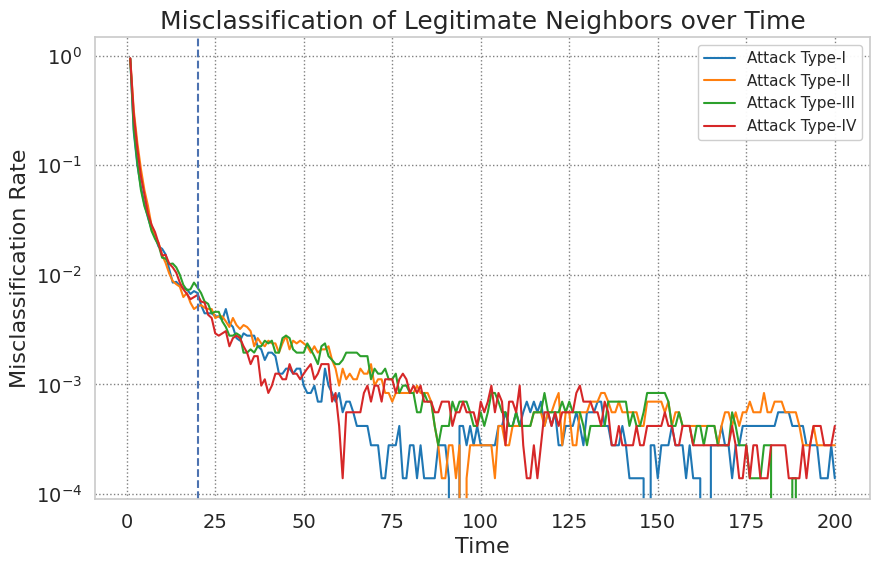}
	\includegraphics[width=.3\linewidth]{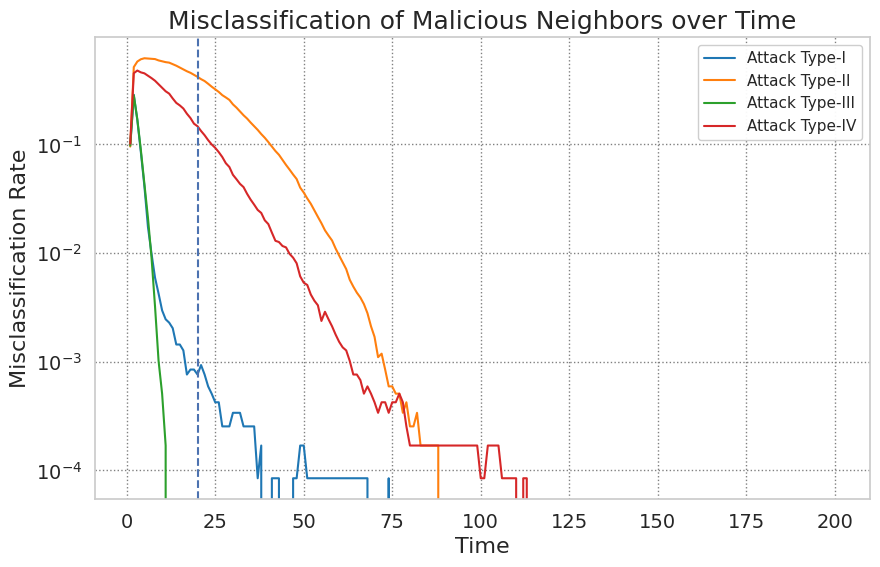}
    	\includegraphics[width=.3\linewidth]{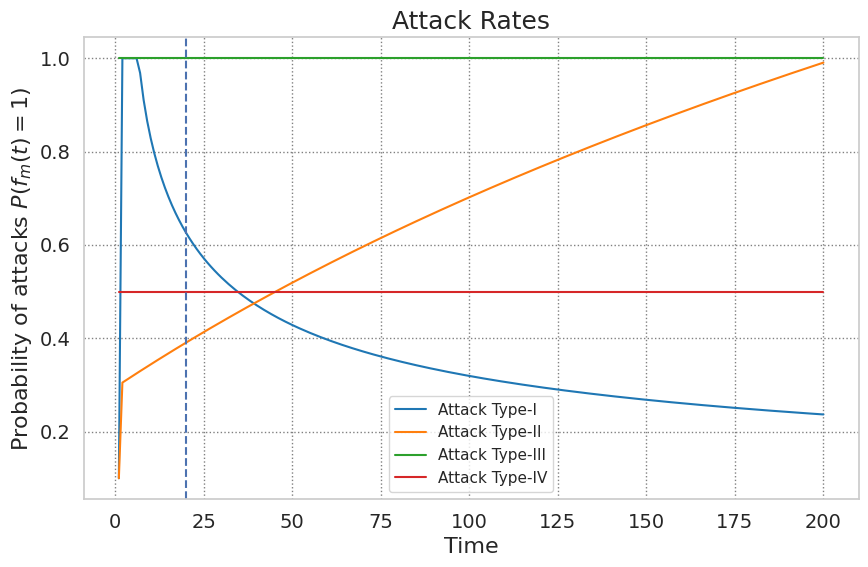}
	\end{tabular}
	\caption{ \footnotesize Average misclassification errors over and attack rates over $100$ runs. (Left) Average misclassification rates of legitimate neighbors (Middle) misclassification rates of malicious neighbors (Right) Average probability of attacks  $\frac{1}{|\ccalM|}\mathbb{P}(f_m(t)=1) $ over time.} \vspace{-1pt}
	\label{fig_mis}
	\vspace{-12pt}
\end{figure*}

\begin{figure}
	\centering
	\begin{tabular}{cc}
	\includegraphics[width=.5\linewidth]{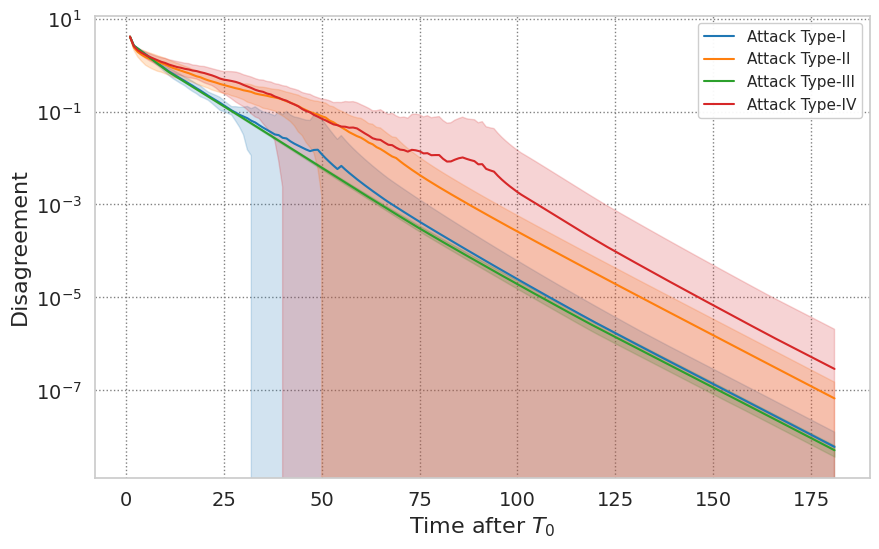}
	\includegraphics[width=.5\linewidth]{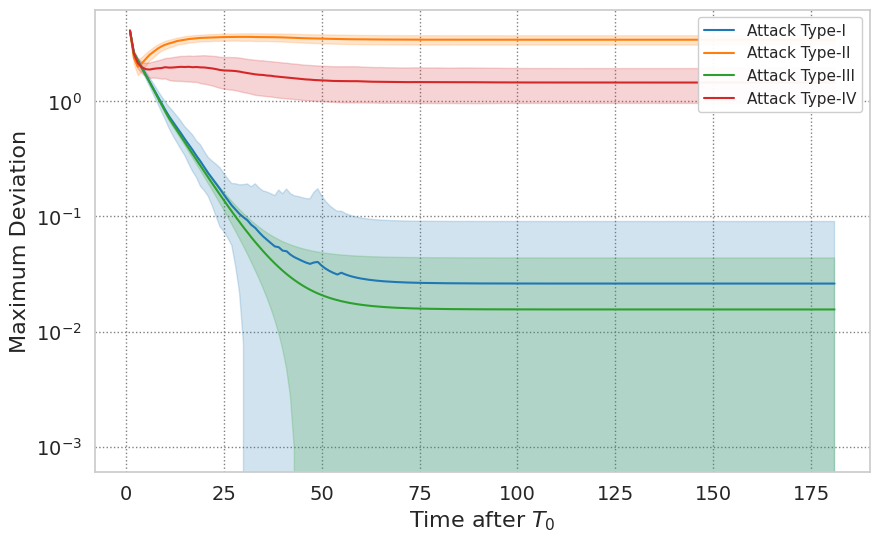}
	\end{tabular}
	\caption{ \footnotesize Trustworthy consensus over $100$ runs. (Left)   Maximum distance to average of agents' values $\max_{i \in \ccalL} |x_i(t)-\frac{1}{|\ccalL|}\sum_{l \in \ccalL} x_l (t)|$ (Right) Maximum deviation from the nominal consensus value. $\max_{i \in \ccalL} |x_i(t)-\mathbf{1} \bar{v} x_\ccalL (0)|
    $, where $\bar{v}$ is the left eigenvector of $\barwl$ corresponding to eigenvalue $1$. } \vspace{-1pt}
	\label{fig_conv}
	\vspace{-12pt}
\end{figure}

In this section, we evaluate the effectiveness of the proposed resilient consensus algorithm in countering different types of malicious attacks using numerical experiments. In this setting, we have $20$ legitimate agents and $30$ malicious agents where the majority of agents are malicious. The communication graph is constructed by first forming a cycle among the legitimate agents, followed by the addition of $20$ pairs of legitimate agents to be assigned edges between them. Malicious agents establish random connections with others, with a probability of $0.2$, ensuring that each is linked to at least one legitimate agent. We generate the communication graph once and keep it fixed during the experiments. The initial values of agents are drawn from the uniform distribution, within the interval $[-4,4]$ ($\eta=4$). Similarly, trust observations for legitimate and malicious transmissions are sampled from uniform distributions with the intervals $[0.4,1]$ and $[0,0.6]$ respectively so that the expected values of transmissions become $E_{\ccalL}=0.7$ and $E_{\ccalM}=0.3$ in order.


We devise four different attack scenarios with time-varying attack probabilities. In the first attack scheme, malicious agents $m \in \ccalM$ only use their former attack history and decide the probability of the next attack $f_m (t) \in \{0,1\}$ using softmax function as follows,

\begin{align} \label{eq_attack_1}
    &\mathbb{P}(f_m(t)=1 \mid \ccalF (t-1))\\
    &=\min \left (p_m(t)+ \exp\left(- r_1 \sum_{k=0}^{t-1} f_m(k)\right),1\right),
\end{align}

where $r_1$ is a constant set to $0.8$, and the sum of lower bounds satisfy $\sum_{k=0}^t p_m(k)= \sqrt{(1+\epsilon_2)(t+1)\ln(t+1)}$ with $\epsilon_2=5$ for all malicious agents $m \in \ccalM$ as per Assumption \ref{asmp_th_lb}. In words, agents reduce their attack rates if they attack more in the past. We correspondingly define the second attack model such that agents increase their attack rates as time increases, 
\begin{align} \label{eq_attack_2}
    \mathbb{P}(f_m(t)=1 \mid \ccalF (t-1))= \min \left (\bar{p}+ \log(1+\exp(-r_2 t )) ,1\right),
\end{align}
where we choose $r_2=0.005$. In both cases the $\min$ functions ensure that the probabilities do not exceed the value $1$. Similarly, in the last model, we assume that malicious agents use independent and identical probabilistic attacks over time with the uniform lower bound $\bar{p}=0.3$ for all time steps $t\ge0$.

In the other two attack models, we consider stationary attack probabilities $\mathbb{P}(f_m(t)=1)=0.5$, and the constant attack model in which agents always attack, implying $\mathbb{P}(f_m(t)=1)=1$ for all times $t\ge0$. Fig. ~\ref{fig_mis} (Right) indicates the average attack probabilities of malicious agents over time. 

In all of the attack models, they send the boundary value $x_m(t)=\eta$ into the consensus process to increase deviation $(f_m(t)=1)$. When they do not attack $(f_m(t)=0)$, they follow a standard consensus process $x_m(t)= w_{mm}x_m(t-1)+\sum_{j \in \mathcal{N}_m} w_{mj} x_j(t-1)$ with the static weights satisfying $w_{mm}>0$, $w_{mj}>0$ and $w_{mm}+ \sum_{j \in \mathcal{N}_m}=1$.

We investigate the four attack scenarios as described in Fig. ~\ref{fig_mis} (Right).  We choose $T_0=25$ as the starting time of the consensus process and threshold parameters $\xi_t = \sqrt{(1+\epsilon_1)(t+1)\ln(t+1)}$ with $\epsilon_1=0.005$ for all $t \ge 0$. Legitimate agents use $\kappa=10$ to form weights as described in Eq. \eqref{eq_wij}.

Fig. ~\ref{fig_mis} exhibits the average misclassification errors and the attack rates. Both of the classification errors converge close to $0$ by the final time $t=200$. In Fig. ~\ref{fig_mis}(Left), the rate of convergence for the misclassification of legitimate agents stays close to each other over time, similar to the conclusion of Lemma \ref{lem_misp_legit}, which provides an upper bound independent of attack rates. Conversely, the role of attack probabilities highly affects the misclassification of malicious agents as in Fig. ~\ref{fig_mis}(Middle). The constants attack rate (green line) is quickly detected around time $t=20$ on average, while malicious agents using the attack models with the constant attack probability $\mathbb{P}(f_m(t)=1)=0.5$ (red line) and with the increasing probabilities (orange line) stay longer in the system on average. The attack model in Eq. \eqref{eq_attack_1} (blue line) has a slowing slope as the attack rates decrease over time. In the comparison of different modes of attacks, the average attack frequency and when to attack with higher rates seem to have a role in the misclassification of malicious agents. Lower values of attack probabilities tend to be detected later. Similarly, if malicious agents have higher attack rates at the beginning, they also have lower rates of misclassification at the beginning.


Fig. ~\ref{fig_conv} shows the convergence performance of the consensus process averaged over $100$ trials. We illustrate the differences between agents' values (Left) and also the deviation from the nominal consensus value over time (Right) on the log scale. Fig. ~\ref{fig_conv} (Left) indicates the agents (nearly) reach consensus around the final time step $t=200$ in all four scenarios.  Fig. ~\ref{fig_conv} (Right) confirms the existence of deviation from the nominal consensus value. Still, the experiments show that the deviation is bounded, and agents' deviations do not fluctuate over time, especially after the time step $t=50$. Further, we see the parallels between Figs. ~\ref{fig_mis} and~\ref{fig_conv}. The attack types detected later on average have a higher impact on the deviation and the consensus error.  This effect is especially observed in Fig.~\ref{fig_conv} (Right) in the case of deviation. However, the differences between the values of legitimate agents still quickly go to $0$ in all cases, which concludes that the proposed consensus dynamics is more robust to these attacks in terms of the value of disagreement. Hence, the numerical experiments align with the theoretical findings (Corollaries \ref{cor_con}-\ref{cor_con_exp} and Theorems \ref{thm_dev}-\ref{thm_rate}). 

\section{Conclusion}
In this paper, we investigated the multi-agent trustworthy consensus problem where agents exchange their values over undirected and static communication networks. Given the availability of stochastic trust observations, we considered the scenarios with dependent sequences of malicious transmissions and trust observations. We established near-geometric decaying misclassification errors using the detection algorithm based on pairwise comparisons of accumulated trust values. This also ensured that after some finite and random time, all agents are correctly classified. Under almost sure correct classification, we also showed that agents reach a consensus almost surely and in expectation asymptotically. For a given probability of failure, we identified the maximal deviation from the nominal consensus process, in terms of the observation window and the number of legitimate and malicious agents, together with the parameters of the detection algorithm. We also derived the convergence rates in finite time. Numerical experiments illustrated the convergence of the consensus process and the deviation under different settings, together with correct classification of agents. 
\bibliographystyle{IEEEtran}
\bibliography{bibliography}
\begin{IEEEbiography}[{\includegraphics[width=1in,height=1.25in,clip,keepaspectratio]{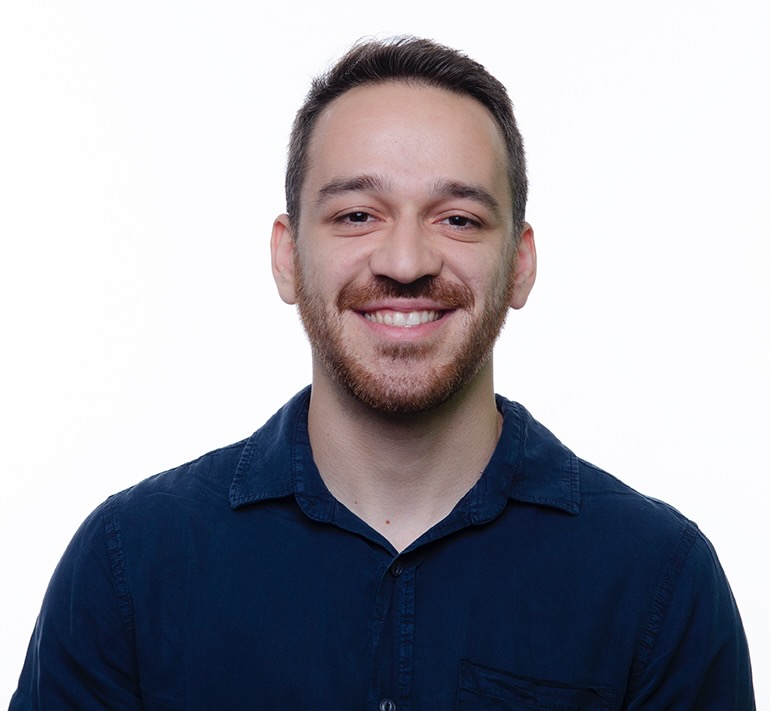}}]{Orhan Eren Akg\"un}
is a Computer Science Ph.D. student in the School of Engineering and Applied Sciences at Harvard University, advised by Prof.~Stephanie Gil. His research is on the development of resilient algorithms to counter adversaries in networked multi-agent systems, specifically within the domain of multi-robot systems. He received his Bachelor’s degree in Electrical and Electronics Engineering from the Bogazici University in 2021.
\end{IEEEbiography}
\begin{IEEEbiography}[{\includegraphics[width=1.05in,height=1.25in,clip,keepaspectratio]{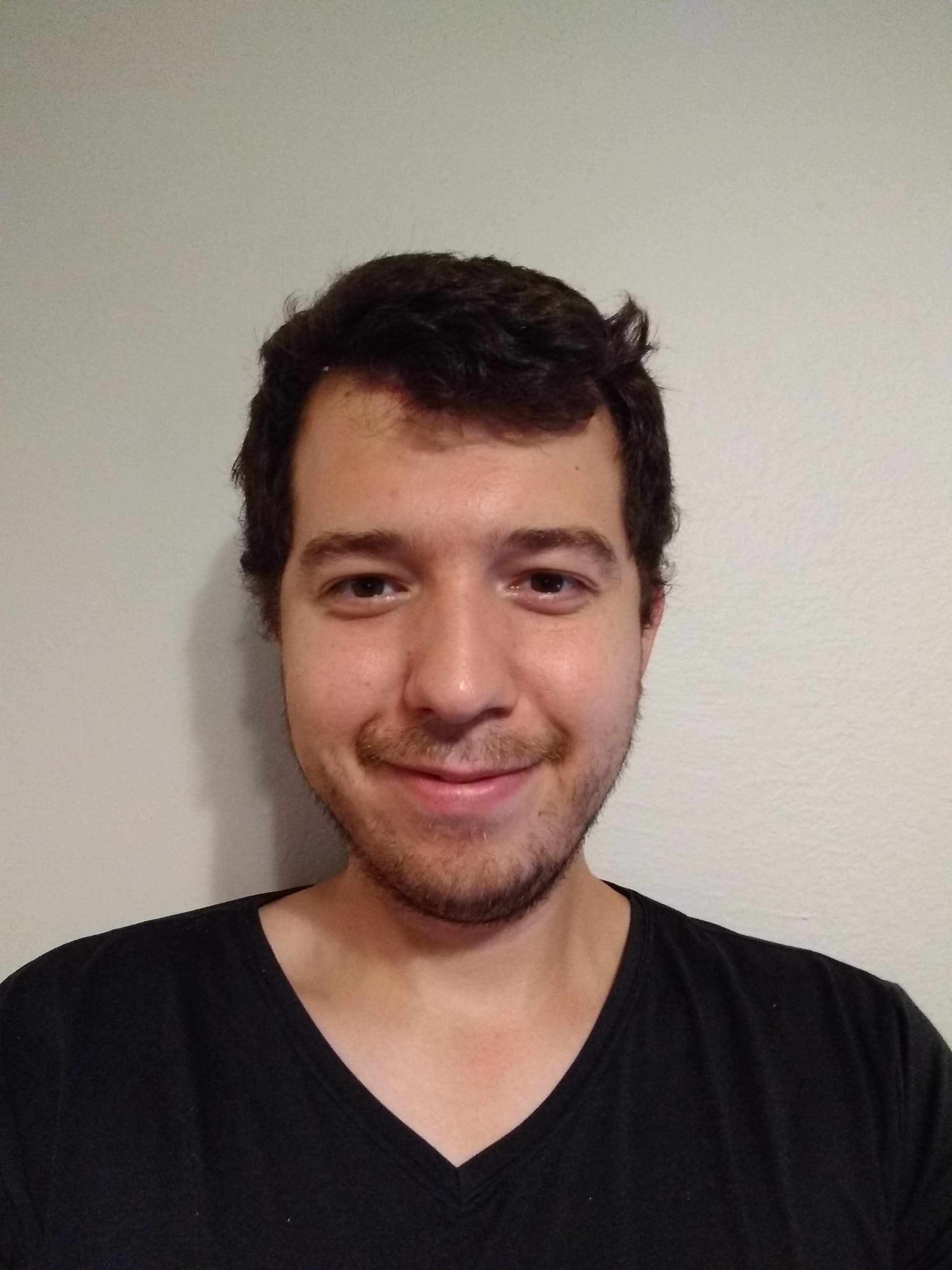}}] {Sarper Ayd\i n } (M'20) earned his B.Sc. degree in Industrial Engineering from Bilkent University, Ankara, Turkey, in 2017. He pursued his Ph.D. studies at Lehigh University, Bethlehem, PA, USA, from 2017 to 2019 before joining Texas A\&M University, where he completed his Ph.D. in Industrial Engineering. He is currently a postdoctoral fellow in the School of Engineering and Applied Sciences at Harvard University. His research focuses on decentralized and resilient algorithms for multi-agent systems.
\end{IEEEbiography}
\begin{IEEEbiography}[{\includegraphics[width=1in,height=1.25in,clip,keepaspectratio]{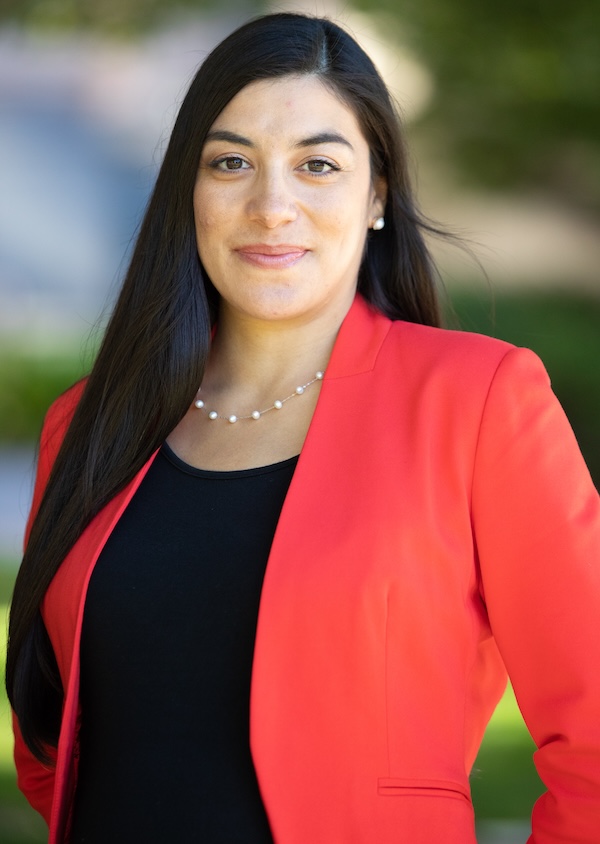}}]{Stephanie Gil}
is an Assistant Professor in the Computer Science Department at the School of Engineering and Applied Sciences at Harvard University where she directs the Robotics, Embedded Autonomy and Communication Theory (REACT) Lab. Prior she was an Assistant Professor at Arizona State University. Her research focuses on multi-robot systems where she studies the impact of information exchange and communication on resilience and trusted coordination. She is the recipient of the 2019 Faculty Early Career Development Program Award from the National Science Foundation (NSF CAREER), the Office of Naval Research Young Investigator Program (ONR YIP) recipient, and has been selected as a 2020 Alfred P. Sloan Fellow.  She obtained her PhD from the Massachusetts Institute of Technology in 2014.
\end{IEEEbiography}
\begin{IEEEbiography}[{\includegraphics[width=1in,height=1.25in,clip,keepaspectratio]{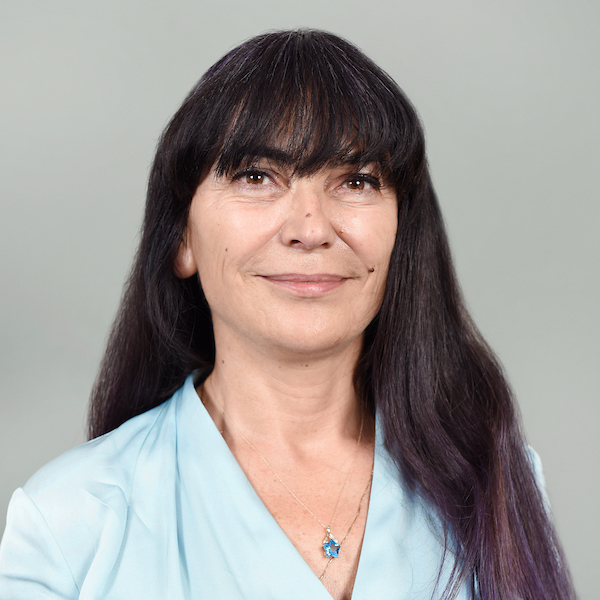}}]{Angelia Nedi\'c} (Member, IEEE),
has Ph.D.\ from Moscow State University, Moscow, Russia, in Computational Mathematics and Mathematical Physics (1994), and Ph.D.\ from Massachusetts Institute of Technology, Cambridge, USA in Electrical and Computer Science Engineering (2002). She has worked as a senior engineer in BAE Systems North America, Advanced Information Technology Division at Burlington, MA. Currently, she is a faculty at the School of Electrical, Computer, and Energy Engineering at Arizona State University at Tempe. Before joining Arizona State University, she was a Willard Scholar faculty member at the University of Illinois at Urbana-Champaign. She is a recipient (jointly with her co-authors) of the Best Paper Award at the Winter Simulation Conference 2013 and the Best Paper Award at the International Symposium on Modeling and Optimization in Mobile, Ad Hoc and Wireless Networks (WiOpt) 2015.  Her general research interest is in optimization, large scale complex systems dynamics, variational inequalities and games.
\end{IEEEbiography}

\end{document}

%% file: my_sections.tex
\usepackage{needspace}



%% file: introduction.tex
In a consensus problem, a set of agents aim to agree on a value using local computations and local interactions over a communication network. The consensus problem is a basis for distributed optimization \cite{NedicOzdaglar,nedic2018distributed}, control \cite{andreasson2014distributed,meng2017distributed}, and estimation \cite{garin2010survey, schizas2007consensus}. In addition, it has also critical relevance for multi-agent coordination applications in cyber-physical systems as it ensures agreement on direction, location, and velocity among agents \cite{kia2019tutorial,bullo2006rendezvous, martinez2009distributed}. Consensus algorithms relying on the full cooperativeness of each agent, are vulnerable to malicious and faulty inputs from agents \cite{pasqualetti2011consensus,sundaram2010distributed}. 

In this paper, we analyze the consensus problem \cite{degroot1974reaching,olfati2004consensus} in the presence of malicious (or faulty) agents whose attack behavior may evolve over time. Specifically, we consider dynamic attack rates, where malicious agents can make attack decisions based on their own history, resulting in potentially dependent and strategic sequences of malicious behavior. This setting captures the possibility of adaptive adversaries that aim to evade detection by selectively choosing when to attack. The premise of this paper is to address and mitigate such dynamic malicious behavior based on detection using quantifiable ``trust" observations as side information derived from the physical aspects of the communication network. 

Resilient consensus methods dealing with malicious agents and data have different methodologies to address the problem. Some studies only utilize transmitted data for detection and elimination. However, they have restrictions on the network connectivity and the total number of malicious agents in the system \cite{dolev1982byzantine,pasqualetti2011consensus,sundaram2010distributed,leblanc2013resilient}. Since these restrictions affect solving other related multi-agent problems, such as optimization~\cite{sundaram2019optimization} and spectrum sensing~\cite{rawat2011sensing}, another body of the literature proposes using additional side information for the assessment of the identity (legitimate vs.\ malicious) of agents \cite{gil2017guaranteeing,cavorsi2023ICRA,Pierson2016,xiong2023securearray}. For example, in a Sybil attack where malicious agents generate imaginary agent identities in a system to have greater influence over the consensus dynamics \cite{gil2019consensus}, or in a location misreporting attack where malicious agents send false data to other agents, the study  \cite{gil2017guaranteeing} details the computation of stochastic trust observations $\alpha_{ij}(t)\in[0,1]$ from wireless signal information, which assess how likely a transmission from a communication link $(i,j)$ at time $t$ is trustworthy or not. In the given attack scenarios, the trust values $\alpha_{ij}(t)$ are derived from checking the uniqueness and directions of wireless signals for Sybil and location misreporting attacks. It is shown in~\cite{gil2017guaranteeing}
that the malicious agents attacking persistently can be detected when
the expected values of trustworthy and malicious transmissions are separated with some constant $\epsilon\in(0,1/2)$, i.e., $\mathbb{E}(\alpha_{ij}(t))\geq 1-\epsilon$ if it is legitimate and $\mathbb{E}(\alpha_{ij}(t))\leq \epsilon$ if it is malicious.

The former work \cite{yemini2021characterizing} proves that malicious agents can be detected via trust observations and the trustworthy agents can reach consensus even in the cases when malicious agents constitute the majority of the total number of agents. The critical property that \cite{yemini2021characterizing} employs is that the value $\epsilon$ is bounded above by $1/2$, which is used as a threshold value to separate the accumulated trust values obtained by summing the trust values $\alpha_{ij}(t)$ over time. However, using a fixed threshold for classification may fail in scenarios where malicious agents exhibit dynamic or random behavior. In fact, in some cases,  intermittently attacking malicious agents can cause more harm to the system than those that attack continuously  \cite{nurellari2018detection,kailkhura2015detection}. The core challenge lies in the fact that the dynamic nature of malicious agents leads to mixed distributions of accumulated trust values, which can closely resemble those of legitimate agents, making them indistinguishable under a fixed threshold detection mechanism based on attack frequency. Consequently, distinguishing attackers from legitimate agents necessitates the use of dynamic (time-varying) thresholds for trust evaluation. In the conference version of this paper \cite{aydin2024multi}, we proposed a new detection algorithm on par with the consensus process resilient against intermittent attacks and failures.  We showed that agents are correctly classified with
probability one if the trust observations for each agent are identical and independently distributed over time. This assumption may not hold in certain scenarios, such as when malicious agents dynamically adjust their attack rates using past available information, creating dependency among trust observations.

In this work, we are motivated by the lack of results that address the dynamic and strategic nature of malicious behavior in the context of resilient consensus dynamics. We extend prior analyses by considering settings where trust observations may be temporally dependent and non-identically distributed, thereby unifying and generalizing existing results on the detection of both intermittent and static malicious behavior. The works \cite{yemini2021characterizing,aydin2024multi}, utilize the concept of \textit{trusted neighborhood} in which legitimate agents choose trustworthy agents to include their data on their updates. Following \cite{aydin2024multi}, we use the Trusted Neighborhood Learning Algorithm (\cref{alg_trust_neig}) to let agents determine their trusted neighborhoods. The algorithm is built on differentiating agents with pairwise comparisons. At first, agents select their most trusted neighbor at each time based on accumulated trust values and then compare it with the remaining neighbor agents. This comparison checks the difference between the accumulated trust values with time-varying thresholds. As a result, agents execute consensus updates only with transmitted data from trusted neighbors. 
In more detail, our contributions in this study are summarized as follows,

    \noindent 1) \textit{Classification and Detection:} Using the detection algorithm (\cref{alg_trust_neig})  we prove that misclassification probabilities decrease (near)-exponentially as time increases (Lemmas \ref{lem_misp_legit}-\ref{lem_misp_mal}). These results show that no misclassification error happens in the trusted neighborhoods after a finite but random time (\cref{lem_als_w}), which we characterize in terms of the difference between expected trust values for malicious and legitimate transmissions, a lower bound on the attack rates, and the parameters of \cref{alg_trust}.\\
    \noindent 2) \textit{Asymptotic  Convergence:}  Relying on almost surely correct classification (\cref{lem_als_w}), we show that legitimate agents reach consensus 
    almost surely (\cref{cor_con}).\\
    \noindent 3) \textit{Deviation from Nominal Consensus:} For a given probability of attack, we characterize the maximal deviation experienced by an agent from the nominal consensus based on the properties of the trust values, the parameters of \cref{alg_trust_neig}, and the numbers of legitimate and malicious agents (Theorem~\ref{thm_dev}).\\
    \noindent 4) \textit{Convergence Rate:} We show that the consensus process converges geometrically fast with a high probability depending on the algorithmic parameters in addition to the numbers of legitimate and malicious agents (\cref{thm_rate}).

\subsection{Related Work}
The consensus problem is well-studied under the conditions of  (strongly)-connected communication networks and (fully) cooperative agents. The previous works \cite{lynch1996distributed,xiao2004fast,olshevsky2009convergence,nedich2015convergence} derive and establish asymptotic convergence properties and convergence rates. Another line of works extends the results for the cases of random failures in communication links \cite{shi2015consensus} and noisy information \cite{touri2009distributed,kar2008distributed}, limited channels \cite{li2010distributed}, and communication delays \cite{liu2018products}. Moreover, the design of weights assigned to other agents in the consensus process has been a subject of interest. The studies consider time-varying weights as a function of (system) states \cite{sluvciak2016consensus}, node degrees \cite{xiao2006distributed}, and negative weights in competitive settings \cite{wu2018consensus}. Overall, these works do not directly address the malicious activity in the consensus problem.

Resilient consensus methods address the presence of malicious agents and focus on mitigating their impact on the system's behavior. As such, the resilient consensus algorithms mainly have two steps \textit{i)} detection/removal of malicious transmissions/agents, and \textit{ii)} consensus update with remaining neighbors/transmissions. The main difference in these studies stems from the issue of the removal of malicious activity. The studies \cite{khalyavin2024non,sundaram2010distributed,leblanc2013resilient} 
utilize Mean Subsequence Reduced (MSR) algorithm (see the review in \cite{ishii2022overview}) to sort incoming data and remove outlier transmissions that are either too large or too small. As extensions of this approach, two recent studies propose new detection methods using information from two-hop neighbors in directed networks \cite{yuan2024resilient}, and a distributed model
predictive control (MPC) for detection of malicious inputs \cite{wei2024resilient}. The disadvantage of these approaches is that they require greater connectivity among legitimate agents, a bounded number of malicious agents, and direct gathering of information from more than just one-hop neighbors. 

In contrast with aforementioned approaches, trust-based methods \cite{yemini2021characterizing,akgun2023learning,hadjicostis2024trustworthy,ballotta2024role} seek to assess the trustworthiness of neighbors with additional trust observations over time rather than solely using transmitted values for detection of anomalies. The works \cite{yemini2021characterizing,akgun2023learning,ballotta2024role} assume implicitly that the set of malicious agents is static and these agents persistently attack. In \cite{hadjicostis2024trustworthy}, similar to \cite{yuan2024resilient}, the algorithm utilizes two-hop neighbor information and further assumes almost surely correct classification of agents without analysis of the behavior of trust observations. The work \cite{fioravanti2023secure} considers resilient gossip algorithm for intermittent malicious attacks. The algorithm uses information from two-hop neighbors and assumes that no agents behave maliciously at the initialization stage, which may not hold when malicious agents attack with dynamic rates. Unlike these works, our approach does not require 
multi-hop information, as it relies only on the availability of trust observations from the immediate neighbors of the agents.


Different concepts of trust have been investigated, such as those where the agents decide on the trustworthiness of other agents using observations \cite{cheng2021general,yang2024enhancing,pippin2014trust}, watermarking \cite{mo2015physical}, sensing \cite{krotofil2015process}, and wireless signals \cite{xiong2023securearray,gil2017guaranteeing,cavorsi2024exploiting}. 
We will use the concept of trust developed in~\cite{gil2017guaranteeing,yemini2021characterizing,cavorsi2024exploiting}. However, unlike these works,
in this paper we are considering sequences of trust observations that are not-necessarily independent. 